\documentclass[fleqn]{lmcs} 
\pdfoutput=1
\usepackage[utf8]{inputenc}

\usepackage{lastpage}
\lmcsdoi{19}{1}{15}
\lmcsheading{}{\pageref{LastPage}}{}{}%
{Mar.~07,~2022}{Mar.~02,~2023}{}


\usepackage{hyperref}
\usepackage[hyphenbreaks]{breakurl}
\usepackage{graphicx}
\usepackage{color}
\usepackage{amsmath,amssymb}
\usepackage{calc,xspace,stmaryrd}
\usepackage{procalg}
\usepackage{tikz}
\usetikzlibrary{arrows,shapes,automata,positioning}
\tikzset{
        state/.style={
          circle,
          draw,
          minimum size=6mm,
    },
}
\usepackage{ulem}
\newcommand{\Procids}[1][\mathcal{P}]{\ensuremath{#1}}

\newcommand{\ProcidsT}{\Procids[\mathcal{P}_{\downarrow}]}
\newcommand{\ProcidsNT}{\Procids[\mathcal{P}_{\!\not\,\downarrow}]}

\newcommand{\ProcidsINT}{\Procids[\overline{\mathcal{P}}_{\!\not\,\downarrow}]}
\newcommand{\Procidssep}{\Procids[\mathcal{P}_{\!\not\,\downarrow}^s]}
\newcommand{\pref}[2][a]{\ensuremath{{#1}.{#2}}}
\newcommand{\seqsym}{\ensuremath{\cdot}}
\newcommand{\seqcsym}{;}
\newcommand{\seq}{\ensuremath{\mathbin{\seqsym}}}
\renewcommand{\seqc}{\ensuremath{\mathbin{\seqcsym}}}
\newcommand{\TSP}{\ensuremath{\mathrm{TSP}}}
\newcommand{\TSPc}{\ensuremath{\mathrm{TSP}^{\seqcsym}}}

\newcommand{\nstep}[1]{\ensuremath{\overset{#1}{\nrightarrow}}}
\newcommand{\nbisim}[1][]{%
    \setbox0=\hbox{\kern-.1ex{$\leftrightarrow$}\kern-.1ex}
    \setbox1=\vbox{\hbox{\raise .1ex \box0}\hrule}%
    \ensuremath{\not\mathrel{\hbox{\kern.1ex\box1\kern.1ex}_{#1}}}
  }
\renewcommand{\nterm}[1]{\ensuremath{#1\mathclose{\!\not\,\downarrow}}}

\newcommand{\sbisim}{\ensuremath{\mathrel{\bisim_s}}}

\newcommand{\defeqn}{\ensuremath{\mathrel{\stackrel{\textrm{def}}{=}}}}

\newcommand{\nasym}{\ensuremath{\mathit{NA}}}
\newcommand{\na}[1]{\ensuremath{\nasym({#1})}}
\newcommand{\NA}[1]{\ensuremath{\mathrm{NA}{#1}}}
\newcommand{\A}[1]{\ensuremath{\mathrm{A}{#1}}}
\newcommand{\C}[1]{\ensuremath{\mathrm{C}{#1}}}
\newcommand{\SI}[1]{\ensuremath{\mathrm{S}{#1}}}

 \newcommand{\rse}[2]
{\ensuremath{{#1}\mathbin{\mkern 1.5mu
\rlap{\raisebox{.5ex}{$\scriptstyle \wedge$}}\mkern 4.5mu
\raisebox{.5ex}{$\scriptscriptstyle \blacktriangle$}{#2}}}}

\begin{document}
 
\title[PDAs and CFGs in bisimulation semantics]{Pushdown Automata and Context-Free Grammars in Bisimulation Semantics\rsuper*}
\titlecomment{{\lsuper*}Extended version of \cite{BCL21}}

\author[J.~C.~M.~Baeten]{Jos C. M. Baeten\lmcsorcid{0000-0003-0287-0555}}[a,b] 
\address{CWI, Amsterdam, The Netherlands}	
\email{Jos.Baeten@cwi.nl}  
\address{University of Amsterdam, Amsterdam, The Netherlands}	

\author[C.~Carissimo]{Cesare Carissimo}[b]	
\email{Jos.Baeten@cwi.nl, cesarecarissimo@gmail.com}  

\author[B.~Luttik]{Bas Luttik\lmcsorcid{0000-0001-6710-8436}}[c]	
\address{Eindhoven University of Technology, Eindhoven, The
  Netherlands}	
\email{s.p.luttik@tue.nl}

\begin{abstract}
The Turing machine models an old-fashioned computer, that does not interact with the user or with other computers, and only does batch processing.
Therefore, we came up with a Reactive Turing Machine that does not have these shortcomings. In the Reactive Turing Machine, transitions have labels to give a notion of interactivity.
In the resulting process graph, we use bisimilarity instead of language equivalence.

Subsequently, we considered other classical theorems and notions from automata theory and formal languages theory. 
In this paper, we consider the classical theorem of the correspondence between pushdown automata and context-free grammars. 
By changing the process operator of sequential composition to a sequencing operator with intermediate acceptance, we get a better correspondence in our setting.
We find that the missing ingredient to recover the full correspondence is the addition of a notion of state awareness.
\end{abstract}

\maketitle

\section{Introduction}

A basic ingredient of any undergraduate curriculum in computer science is a course on automata theory and formal languages, as this gives students insight in the essence of a computer, and tells them what a computer can and cannot do. Usually, such a course contains the treatment of the Turing machine as an abstract model of a computer. However, the Turing machine is a very old-fashioned computer: it is deaf, dumb and blind, and all input from the user has to be put on the tape before the start. Computers behaved like this until the advent of the terminal in the mid 1970s. This is far removed from computers the students find all around them, which interact continuously with people, other computers and the internet. It is hard to imagine a self-driving car driven by a Turing machine that is deaf, dumb and blind, where all user input must be on the tape at the start of the trip.

In order to make the Turing machine more interactive, many authors have enhanced it with extra features, see e.g. \cite{goldin2008interactive,wegner1997interaction}. But an extra feature, we believe, is not the way to go. Interaction is an essential ingredient, such as it has been treated in many forms of concurrency theory. We seek a full integration of automata theory and concurrency theory, and proposed the Reactive Turing Machine in \cite{BLT13}. 
In the Reactive Turing Machine, transitions have labels to give a notion of interactivity.
In the resulting process graphs, we use bisimilarity instead of
language equivalence. Subsequently, we considered other classical
theorems and notions from automata theory and formal languages theory
\cite{BCLT09,BLMvT16}. We find richer results and a finer theory.

In this paper, we consider the classical theorem of the correspondence
between pushdown automata and context-free grammars: a formal
  language is accepted by a pushdown automaton if, and only if, it is
  generated by a context-free grammar.
Our aim is to establish a process-theoretic variant of that correspondence theorem: a process is defined by a pushdown automaton
  if, and only if, can be specified in a process algebra comprising
  actions, choice, sequencing and recursion. There are several choices
that we need to make both regarding the process-theoretic semantics of
pushdown automata and regarding the process algebra.

  Regarding the process-theoretic semantics of pushdown automata, in
  contrast to language semantics, in bisimulation semantics it makes a
  difference whether we consider acceptance by final state or
  acceptance by empty stack. Every process that is a pushdown process
  according to the acceptance-by-empty-stack interpretation is also a
  pushdown process according to the acceptance-by-final-state
  interpretation, but the converse does not hold \cite{BCLT09}. In
  this paper, we focus on the  acceptance-by-final-state
  interpretation and set out to find the corresponding process algebra.

  We start out from the seminal process algebra BPA of Bergstra and Klop \cite{BK84}. Then, we need to add constants for acceptance and non-acceptance (deadlock), so we look at the process algebra \TSP{} of \cite{BBR10}. However, as we found earlier (starting from \cite{BCT08}), the process algebra \TSP{} is not a good choice for the correspondence with pushdown auitomata modulo bisimulation: by means of a guarded recursive specification over \TSP{} we can define processes that cannot be generated by any pushdown automaton, and some processes generated by a pushdown automaton cannot be defined by a guarded recursive specification over \TSP{}.

By changing the process operator of sequential composition to a sequencing operator with intermediate acceptance, we get the theory \TSPc{} that provides a better correspondence in our setting \cite{BLY17,Bel18,BLB19}. All processes defined by a guarded recursive specification over \TSPc{} can also be generated by a pushdown automaton. However, the other direction still does not hold.

We find that the missing ingredient to recover the full correspondence
is the addition of a notion of state awareness, by means of a signal
that can be passed along a sequencing operator.

 Our main
contribution is to show that it suffices to add propositional signals
and conditions in the style of \cite{BB97} to restore the correspondence: we establish
that a process is specified by a guarded sequential recursive
specification with propositional signals and conditions if, and only
if, it is the process associated with a pushdown automaton. 

This paper extends \cite{BCL21} by adding an axiomatisation of the
process theory. This axiomatisation allows us to derive the head
normal theorem already reported in \cite{BCL21} purely equationally
and is proved to be ground-complete for the recursion-free fragment of
the process theory.

\section{Preliminaries}

As a common semantic framework we use the notion of a
\emph{labelled transition system}.

\begin{defi} \label{def:tsspace}
A \emph{labelled transition system} is a quadruple
$(\mathcal{S},\mathcal{A},{\xrightarrow{}},{\downarrow})$, where
\begin{enumerate}
\item $\mathcal{S}$ is a set of \emph{states};
\item $\mathcal{A}$ is a set of \emph{actions}, $\tau\not\in\mathcal{A}$;
\item
${\xrightarrow{}}\subseteq{\mathcal{S}\times\mathcal{A}\cup\{\tau\}\times\mathcal{S}}$ is
an $\mathcal{A}\cup\{\tau\}$-labelled \emph{transition relation}; and
     \item ${\downarrow}\subseteq\mathcal{S}$ is the set of \emph{final} or
       \emph{accepting} states.
 \end{enumerate}
A \emph{process graph} is a
labelled transition system with a special
designated \emph{root state} ${\uparrow}$, i.e., it is a quintuple
$(\mathcal{S},\mathcal{A},{\rightarrow},{\uparrow},{\downarrow})$ such that
$(\mathcal{S},\mathcal{A},{\rightarrow},{\downarrow})$ is a labelled transition system, and ${\uparrow}\in\mathcal{S}$.
We write $s\xrightarrow{a}s'$ for $(s,a,s')\in{\rightarrow}$
and $\term{s}$ for $s\in\mathalpha{\downarrow}$.
\end{defi}

By considering language equivalence classes of process graphs, we
recover languages as a semantics, but we can also consider other
equivalence relations. Notable among these is \emph{bisimilarity}.

\begin{defi}
  Let $(\mathcal{S},\mathcal{A},\rightarrow{},{\downarrow})$ be a
  labelled transition system. A symmetric binary relation $R$ on
  $\mathcal{S}$ is a \emph{bisimulation} if it satisfies the following
  conditions for every $s,t\in\mathcal{S}$ such that $s\mathrel{R} t$ and for all $a\in\mathcal{A}\cup\{\tau\}$:
  \begin{enumerate}
    \item if $s\xrightarrow{a}s'$ for some $s'\in\mathcal{S}$,
      then there is a $t'\in\mathcal{S}$ such that
      $t\xrightarrow{a}t'$ and $s'\mathrel{R}t'$; and
    \item if $s{\downarrow}$, then $t{\downarrow}$.
    \end{enumerate}
\end{defi}
The results of this paper do not rely on abstraction from internal
computations. This means we can use the \emph{strong} version of bisimilarity
defined above, which does not give special treatment to
$\tau$-labelled transitions. In general, when we do give special treatment to $\tau$-labeled transitions, we use some form of \emph{branching bisimulation} \cite{GW96}.

A \emph{process} is a (strong) bisimulation equivalence class of process graphs.

\section{Pushdown Automata}

We consider an abstract model of a computer with a memory in the form of a \emph{stack}: the stack is last in and first out, something can be added on top of the stack (push\index{push}), or something can be removed from the top of the stack (pop\index{pop}).

 \begin{defi}
 A \emph{pushdown automaton} $M$ is a sextuple $(\mathcal{S},\mathcal{A}, \mathcal{D},{\rightarrow},{\uparrow},{\downarrow})$ where:
 \begin{enumerate}
 \item $\mathcal{S}$ is a finite set of states,
 \item $\mathcal{A}$ is a finite input alphabet, $\tau \not\in \mathcal{A}$ is the unobservable step,
 \item $\mathcal{D}$ is a finite data alphabet,
 \item $\mathalpha{\rightarrow} \subseteq \mathcal{S} \times  (\mathcal{A} \cup \{\tau\}) \times (\mathcal{D} \cup \{\epsilon\}) \times \mathcal{D}^{*} \times \mathcal{S}$ is a finite set of \emph{transitions} or \emph{steps},
 \item $\mathalpha{\uparrow} \in \mathcal{S}$ is the initial state,
 \item $\mathalpha{\downarrow} \subseteq \mathcal{S}$ is the set of final or accepting states.
 \end{enumerate}
\end{defi}

 If $(s,a,d,x,t) \in \mathalpha{\rightarrow}$ with $d \in \mathcal{D}$, we write $s \xrightarrow{a[d/x]} t$, and this means that the machine, when it is in state $s$ and $d$ is the top element of the stack, can consume input symbol $a$, replace $d$ by the string $x$ and thereby move to state $t$. Likewise, writing $s \xrightarrow{a[\epsilon/x]} t$ means that the machine, when it is in state $s$ and the stack is empty, can consume input symbol $a$, put the string $x$ on the stack and thereby move to state $t$. 
 In steps $s \xrightarrow{\tau[d/x]} t$ and $s
 \xrightarrow{\tau[\epsilon/x]} t$, no input symbol is consumed, only
 the stack is modified.
 
 Notice that we defined a pushdown automaton in such a way that it can be detected if the stack is empty, i.e. when there is no top element.

\begin{figure}[htb]
\begin{center}
\begin{tikzpicture}[->,>=stealth',node distance=2cm, node font=\footnotesize, state/.style={circle, draw, minimum size=.5cm,inner sep=0pt}]
  \node[state,initial,initial text={},initial where=left,accepting] (s0) {$\uparrow$};
  
  \path[->]
  (s0) edge[in=-45,out=45,loop]
         node[right] {$\begin{array}{c}a[\epsilon/1]\\
                                      a[1/11]\\
                                      b[1/\epsilon]\end{array}$} (s0);
\end{tikzpicture}
\end{center}
\caption{Pushdown automaton of a counter.}\label{fig:pda}
\end{figure}

For example, consider the pushdown automaton
depicted in Figure~\ref{fig:pda}. It represents the process that can
start to read an $a$, and after it has read at least one $a$, can read additional $a$'s but can also
read $b$'s. Upon acceptance, it will have read up to as many
$b$'s as it has read $a$'s.
Interpreting symbol $a$ as an increment and $b$ as a decrement, we can see this process as a \emph{counter}.
   
We do not consider the language of a pushdown automaton, but rather
consider the process, i.e., the bisimulation equivalence class of the
process graph of a pushdown automaton. A
state of this process graph is  a pair
$(s,x)$, where $s \in \mathcal{S}$ is the current state and $x \in \mathcal{D}^{*}$ is the current contents of the stack (the left-most element
of $x$ being the top of the stack). In the initial state, the stack is
empty. In a final state, acceptance can take place irrespective of the
contents of the stack. The transitions in the process graph are labeled by the inputs of the pushdown automaton or $\tau$.

\begin{defi}\label{def:pdalts}
  Let
    $M=(\mathcal{S},\mathcal{A}, \mathcal{D},{\rightarrow},{\uparrow},{\downarrow})$
  be a pushdown automaton.
  The \emph{process graph}
  $\mathcal{P}(M)=(\mathcal{S}_{\mathcal{P}(M)},\mathcal{A},{\xrightarrow{}}_{\mathcal{P}(M)},{\uparrow}_{\mathcal{P}(M)},{\downarrow}_{\mathcal{P}(M)})$ associated with $M$ is
  defined as follows:
  \begin{enumerate}
  \item $\mathcal{S}_{\mathcal{P}(M)} = \{(s,x)\mid s\in\mathcal{S}\ \&\
    x\in\mathcal{D}^{*}\}$;
  \item ${\xrightarrow{} _{\mathcal{P}(M)}}\subseteq
    {\mathcal{S}_{\mathcal{P}(M)}\times\mathcal{A}\cup\{\tau\}\times
      \mathcal{S}_{\mathcal{P}(M)}}$ is the least relation such that
    for all $s,s'\in\mathcal{S}$, $a\in\mathcal{A}\cup\{\tau\}$, $d\in\mathcal{D}$ and
    $x,x'\in\mathcal{D}^{*}$ we have
    \begin{equation*}
      (s,dx)\xrightarrow{a}_{\mathcal{P}(M)}(s',x'x)\
        \text{if, and only if,}\ s\xrightarrow{a[d/x']}s'\enskip;
      \end{equation*}
         \begin{equation*}
      (s,\epsilon)\xrightarrow{a}_{\mathcal{P}(M)}(s',x)\
        \text{if, and only if,}\ s\xrightarrow{a[\epsilon/x]}s'\enskip;
      \end{equation*}
  \item $\uparrow_{\mathcal{P}(M)}=(\uparrow,\epsilon)$;
 \item ${\downarrow}_{\mathcal{P}(M)}= \{(s,x)\mid s\in{\downarrow}\
    \&\ x\in\mathcal{D}^{*}\}$.
   \end{enumerate}
\end{defi}

To distinguish, in the definition above, the set of states, the
transition relation, the initial state and the set of accepting states
of the pushdown automaton from similar components of the associated
process graph, we have attached a subscript ${\mathcal{P}(M)}$ to
  the latter. In the remainder of this paper, we will suppress the
  subscript whenever it is already clear from the context whether a
  component of the pushdown automaton or its associated process graph is meant.

\begin{figure}[htb]
\begin{center}
\begin{tikzpicture}[->,>=stealth',node distance=2cm, node
  font=\footnotesize, state/.style={ellipse, draw, minimum size=.5cm,inner sep=0pt}]
  \node[state,accepting,initial,initial text={},initial where=left] (s0)
  {$(\uparrow,\epsilon)$};
  \node[state,accepting] [right of=s0] (s1) {$(\uparrow,1)$};
  \node[state,accepting] [right of=s1] (s2) {$(\uparrow,11)$};
  \node[state,accepting, draw=none] [right of=s2] (sdots) {$\dots$};

  \path[->]
    (s0) edge[bend left] node[above] {$a$} (s1)
    (s1) edge[bend left] node[above] {$a$} (s2)
    (s2) edge[bend left] node[above] {$a$} (sdots);
  \path[->]
    (sdots) edge[bend left] node[below] {$b$} (s2)
    (s2) edge[bend left] node[below] {$b$} (s1)
    (s1) edge[bend left] node[below] {$b$} (s0);
\end{tikzpicture}
\end{center}
\caption{The process graph of the counter.}\label{fig:pdatrans}
\end{figure}

Figure~\ref{fig:pdatrans} depicts the process graph associated
with the pushdown automaton depicted in Figure~\ref{fig:pda}.

In language equivalence, the definition of acceptance in pushdown automata leads to the same set of languages when we define acceptance by final state (as we do here) and when we define acceptance by empty stack (not considering final states). In bisimilarity, these notions are different: acceptance by empty stack yields a smaller set of processes than acceptance by final state. This is illustrated by Figure~\ref{fig:pdatrans}: this process graph has infinitely many non-bisimilar final states. This cannot be realised if we define acceptance by empty stack. For details, see \cite{BCLT09}.

A pushdown automaton has only finitely many transitions, so there is a maximum number of transitions from a given state, called its \emph{branching degree}. Then, also the associated process graph has a branching degree, that cannot be larger than the branching degree of the underlying pushdown automaton. Thus, in a process graph associated with a pushdown automaton, the branching is always \emph{bounded}.

 \section{Sequential Processes: \TSP}

In the process setting, a context-free grammar is a recursive
specification over a process algebra comprising actions, choice and
sequencing. We start out from the seminal process algebra BPA of
Bergstra and Klop \cite{BK84}. Later, it was extended to the process
algebra BPA$_{\delta \varepsilon}$ by adding two constants $\delta$
(for deadlock) and $\varepsilon$ (for termination). This process
algebra was reformulated as the Theory of Sequential Processes (\TSP)
in \cite{BBR10}, using  $\dl$ and $\emp$ for $\delta$ and
$\varepsilon$. We present \TSP{} next, and then will argue that it is
not suitable to obtain a correspondence with pushdown automata.

Let $\Act$ be a set of \emph{actions} and $\tau\not\in\Act$ \emph{the silent action}, symbols denoting atomic
events, and let $\Procids$ be a finite set of \emph{process identifiers}. The
sets $\Act$ and $\Procids$ serve as parameters of the process theory
that we shall introduce below. We use symbols $a,b,\ldots$, possibly indexed, to range over $\Act \cup \{\tau\}$, symbols $X,Y,\dots$, possibly indexed, to range over $\Procids$.
The set of \emph{sequential process expressions} is generated by the following grammar ($a\in\Act \cup \{\tau\}$, $X\in\Procids$):
\begin{equation*}
p ::= \dl \mid \emp \mid \pref[a]{p} \mid p + p \mid p \seq  p \mid X\enskip.
\end{equation*}
The constants $\dl$ and $\emp$ respectively denote the
\emph{deadlocked} (i.e., inactive but not accepting)
process and the \emph{accepting} process. For each
$a\in\Act \cup \{\tau\}$ there is a unary action prefix operator
$\pref[a]{\_}$. The binary operators $+$ and $\seq$ denote
alternative composition and sequential composition,
respectively. We adopt the convention that $\pref[a]{\_}$ binds
strongest and $+$ binds weakest.
The symbol $\cdot$ is often omitted when writing sequential process
expressions.

For a (possibly empty) sequence
$p_1,\dots,p_n$ we inductively define $\sum_{i=1}^np_i=\dl$ if $n=0$
and $\sum_{i=1}^{n}p_i=(\sum_{i=1}^{n-1}p_i)+p_n$ if $n>0$.

A recursive specification over sequential process expressions is a mapping
$\Delta$ from $\Procids$ to the set of sequential process expressions. The idea is that the process expression
$p$ associated with a process identifier $X\in\Procids$ by $\Delta$
\emph{defines} the behaviour of $X$. We prefer to think of $\Delta$ as a
collection of \emph{defining equations}
  $X\defeqn p$,
  exactly one for every $X\in\Procids$.
We shall, throughout the paper, presuppose a recursive specification
$\Delta$ defining the process identifiers in $\Procids$, and we shall
usually simply write $X\defeqn p$ for $\Delta(X)=p$. Note that, by our
assumption that $\Procids$ is finite, $\Delta$ is finite too.

\begin{figure}[htb]
  \centering
  \begin{osrules}
        \osrule*{}{\emp \downarrow}
        \qquad \qquad
        \osrule*{}{\pref[a]p\step{a}p}
    \\
\osrule*{p \downarrow}{(p+q) \downarrow}
\quad
\osrule*{q \downarrow}{(p+q) \downarrow}
\quad
    \osrule*{p \step{a} p'}{p + q \step{a} p'}
    \quad
    \osrule*{q \step{a} q'}{p + q \step{a} q'}
  \\
 \osrule*{p \downarrow & q \downarrow}{p \seq  q \downarrow}\label{osrule:seqone}
  \qquad
    \osrule*{p\step{a} p'}{p \seq  q \step{a} p' \seq  q} \label{osrule:seqtwo}
  \qquad
    \osrule*{p \downarrow &  q \step{a} q' }{p \seq  q
      \step{a} q'}\label{osrule:seqthree}
    \\
    \osrule*{p\step{a}p' & X\defeqn p}{X\step{a}p'}
  \qquad
  \osrule*{\term{p} & X\defeqn p}{\term{X}}
  \end{osrules}
\caption{Operational semantics for sequential process expressions.}
\label{fig:semantics_tspseq}
\end{figure}

We associate behaviour with process expressions by defining, on the
set of process expressions, a unary acceptance predicate $\term{}$
(written postfix) and, for every $a\in\Act \cup \{\tau\}$, a binary transition
relation $\step{a}$ (written infix), by means of the transition system
specification presented in Figure~\ref{fig:semantics_tspseq}. We write
$p\nstep{a}$ for ``there does not exist $p'$ such that $p\step{a}p'$''
and $p\nstep{}$ for ``$p\nstep{a}$ for all $a\in\Act \cup \{\tau\}$''. 

For $w\in\Act^{*}$ we define
$p\steps{w}p'$ inductively, for all process expressions $p,p',p''$;
\begin{itemize}
\item $p \steps{\epsilon} p$;
\item if $p \step{a} p'$ and $p' \steps{w} p''$, then $p \steps{aw} p''$ ($a \in \Act$);
\item if $p \step{\tau} p'$ and $p' \steps{w} p''$, then $p \steps{w} p''$.
\end{itemize}

We see that $\tau$-steps do not contribute to the string $w$.
We write
$p\step{}p'$ for there exists $a\in\Act \cup \{\tau\}$ such that
$p\step{a}p'$. Similarly, we write $p\steps{}p'$ for there exists
$w\in\Act^{*}$ such that $p\steps{w}p'$ and say that $p'$ is
\emph{reachable} from $p$.

When a process expression $p$ satisfies both $p \downarrow$ and $p
\step{}$ we say $p$ has \emph{intermediate acceptance}. We will need
to take special care of such process expressions in the sequel.

We proceed to define when two process expressions are behaviourally equivalent.

\begin{defi} \label{def:bisimilarity}
A binary relation $R$ on the set of sequential process expressions is a \emph{bisimulation} iff $R$ is
symmetric and for all process expressions $p$ and $q$ such that if $(p,q) \in
R$:
\begin{enumerate}
\item If $p \step{a} p'$, then there exists a process expression $q'$, such that $q \step{a} q'$, and $(p', q') \in R$.
\item If $\term{p}$, then $\term{q}$.
\end{enumerate}
The process expressions $p$ and $q$ are bisimilar (notation: $p \bisim q$) iff there exists a bisimulation $R$ such that $(p, q) \in R$.
\end{defi}

The operational rules presented in Fig~\ref{fig:semantics_tspseq} are
in the so-called \emph{path format} from which it immediately follows
that strong bisimilarity is a congruence \cite{BV93}. Branching bisimulation, however, is not a congruence, but by adding a rootedness condition we get rooted branching bisimulation which is a congruence \cite{GW96}.

As is customary in process theory, we restrict our attention  to \emph{guarded}
recursive specifications, i.e., we require that every occurrence of a
process identifier in the definition of some (possibly different)
process identifier occurs within the scope of an action prefix. Note
that we allow $\tau$ as a guard. This is possible since we use strong
bisimulation, not rooted branching bisimulation.

\begin{figure}
\begin{tikzpicture}[->,>=stealth',node distance=2cm, node font=\footnotesize, state/.style={ellipse, draw, minimum size=0.5cm,inner sep=0pt}]
  \node[initial, initial text=,state, double, double distance=1pt] (a1) at (1,2) {$X$};
  \node (a5) at (7,1) {};
  \node[state, double, double distance=1pt] (b1) at (1,0) {$(a.\emp)^0$};
  \node[state] (b2) at (3,0) {$(a.\emp)^1$};
  \node[state] (b3) at (5,0) {$(a.\emp)^2$};
  \node[state] (b4) at (7,0) {$(a.\emp)^3$};
  \node (b5) at (9,0) {};
  \path[->] (a1) edge [right] node {$a$} (b1);
  \path[->] (a1) edge [right]  node[xshift=0.1cm] {$a$} (b2);
  \path[->] (a1) edge [right]  node[xshift=0.1cm] {$a$} (b3);
  \path[->] (a1) edge [right]  node[xshift=0.2cm] {$a$} (b4);
  \path[->] (b4) edge [below] node {$a$} (b3);
  \path[->] (b3) edge [below] node {$a$} (b2);
  \path[->] (b2) edge [below] node {$a$} (b1);
  \path[->, dashed] (a1) edge [right] node[xshift=0.2cm] {$a$} (a5);
  \path[->, dashed] (b5) edge [below] node {$a$} (b4);
\end{tikzpicture}
\caption{The process graph of $X$ as defined in $X \defeqn \emp + X \seq
  a.\emp$; the terms $(a.\emp)^n$ are inductively defined by $(a.\emp)^0=\emp$
  and $(a.\emp)^{n+1} = (a.\emp)^n \seq (a.\emp)$.}\label{fig:infbranch}
\end{figure}

If we do not restrict to guarded recursion, unwanted behaviour can result, as the following example shows.
\begin{exa} \label{exa:unguarded}
Consider the (unguarded) recursive equation
\[ X \defeqn \emp + X \seq a.\emp \enskip. \]
We show the process graph generated by the operational rules in Figure~\ref{fig:infbranch}.

From the initial state, there are infinitely many transitions, to states that are all non-bisimilar. Thus, the process graph generated is infinitely branching, and it cannot be bisimilar to the graph of a pushdown automaton.
\end{exa}

By restricting to guarded recursion, we guarantee that the associated
process graphs are finitely branching.

It is convenient to present a guarded sequential
specification in a normal form, the so-called Greibach normal form,
see \cite{BLB19}.

\begin{defi} \label{GNF}
A guarded sequential specification is in Greibach normal form, GNF, if every right-hand side of every equation has the following form:
\begin{itemize}
\item $(\emp +) \sum_{i=1}^{n} a_i.\alpha_i$ for actions $a_i \in
  \mathcal{A} \cup \{\tau\}$ and a sequence of identifiers $\alpha_i \in \mathcal{P}^{*}$, $n \geq 0$.
\end{itemize}
Recall that the empty summation equals $\dl$. The $\emp$ summand may or may not occur. Furthermore, the empty sequence denotes $\emp$.
\end{defi}

It is well-known that by adding a finite number of process identifiers, every guarded sequential specification can be brought into Greibach normal form (i.e. the behaviour associated with a process identifier by the original specification is strongly bisimilar to
the behaviour associated with it by the transformed specification,
see, e.g.,  \cite{BLB19}).

Since bisimilarity is a congruence, we can consider the equational
theory of sequential expressions. From \cite{BBR10}, we know that the finite axiomatization in Table~\ref{tab:axioms_tsp} is
a sound and ground-complete (i.e. complete for all process
expressions, not including variables) axiomatisation for the theory TSP
with sequential composition and without recursion. In axiomatisations, we use variables $x,y,z$ denoting arbitrary process expressions.

\begin{table}[htb]

\centering
\begin{tabular}{lcl@{\qquad}l}
$x+y$ & $=$ & $y+x$ & $\A1$ \\
$x+(y+z)$ & $=$& $(x+y)+z$ & $\A2$ \\
$x+x$ & $=$ & $x$ & $\A3$ \\ 
$(x+y) \seq z$ & $=$ & $x \seq z + y \seq z$ & $\A4$\\
$(x \seq y) \seq z$ & $=$ & $x  \seq (y \seq z)$ & $\A5$ \\
$x+\dl$ & $=$ & $x$ & $\A6$ \\
  $\dl \seq x$ & $=$ & $\dl$ & $\A7$ \\
$x \seq \emp$ & $=$ & $x$ & $\A8$ \\
$\emp \seq x$ & $=$ & $x$ & $\A9$ \\
$(\pref[a]x)\seq y$ & $=$ & $\pref[a](x \seq y)$ & $\A{10}$ \\ \\
\end{tabular}

\caption{The axioms of TSP ($a \in \Act \cup \{\tau\}$).}\label{tab:axioms_tsp}
\end{table}

Now we show why the process algebra TSP is not suitable to establish a correspondence with pushdown automata. The problem is with the operational semantics of sequential composition.
Consider the following example.

\begin{figure}[h]
  \centering
\begin{tikzpicture}[->,>=stealth',node distance=2cm, node font=\footnotesize, state/.style={ellipse, draw, minimum size=0.5cm,inner sep=0pt}]

\node[initial, initial text=, state,double, double distance=1pt]	(7)        {$X$};
\node[state,double, double distance=1pt]	(8)	[right of=7]	    {$YX$};
\node[state,double, double distance=1pt]	(9)	[right of=8]     	{$Y^2X$};
\node[state,double, double distance=1pt]	(10)[right of=9]     	{$Y^{n-1}X$};
\node[state,double, double distance=1pt]	(11)[right of=10]    	{$Y^nX$};
\node[state,draw=none]	(12)[right of=11]    	{};

\path
    (7) edge [above, bend left] node {$a$}	(8)
		(8) edge [above, bend left] node {$a$}	(9)
		(9) edge [above,dashed, bend left] node {}	(10)
		(10) edge [above,bend left] node {$a$}	(11)
		(11) edge [above,dashed,bend left] node	{}	(12)
			
		(12) edge [below,dashed,bend left] node	{}	(11)
		(11) edge [below,bend left] node			{$b$}	(10)
		(10) edge [below,dashed,bend left] node	{}	(9)
		(9)  edge [below,bend left] node {$b$}	(8)
		(8)  edge [below,bend left] node {$b$}	(7)
		
		(11.south) edge [below,bend left=40] node[xshift=-.4cm,yshift=.1cm] {$b$}	(9.south)
		(11.south) edge [below,bend left=45] node {$b$}	(8.south)
		(11.south) edge [below,bend left=48] node {$b$}	(7.south)
		
		(10.south) edge [below,bend left=40] node[yshift=.1cm] {$b$}	(8.south)
		(10.south) edge [below,bend left=45] node {$b$}	(7.south)
		
		(9.south) edge [below,bend left=40] node[xshift=.4cm,yshift=.1cm] {$b$}	(7.south)
;
\end{tikzpicture}
\caption{Process graph showing unbounded branching.}
\label{fig:UnboundedBranching}
\end{figure}

\begin{exa} \label{exa:difference}
  Consider the recursive specification
  \begin{equation*}
    X \defeqn \emp + \pref[a]{Y \seq X} \qquad
    Y \defeqn \emp + \pref[b]{\emp}+ a.Y \seq Y \enskip.
  \end{equation*}
  By following the operational rules, we obtain a process graph that
  is bisimilar to the one shown in
  Figure~\ref{fig:UnboundedBranching} (for simplicity we have
  identified states labelled with terms that are equal up to \A5, \A8
  and \A9). We see the similarity with Figure~\ref{fig:pdatrans}, but many more transitions are added.

We see that the state given by process expression $Y^n \seq X$ has $n$ outgoing transitions, to states $X, Y \seq X, \ldots, Y^{n-1} \seq X, Y^{n+1} \seq X$. These states are all non-bisimilar. Thus, the branching in this process graph is unbounded, and it cannot be the process graph of a pushdown automaton.
\end{exa}

We call the phenomenon occurring here \emph{transparency}: it is possible to skip part of the terms in a sequential composition. It causes that \TSP{} is too expressive: by means of a guarded recursive specification we can specify a process that is not pushdown. On the other hand,  \TSP{} is insufficiently expressive at the same time: also due to transparency, the pushdown automaton in Figure~\ref{fig:pda} cannot be specified in \TSP.

\begin{thm}
There is no guarded recursive specification over \TSP{} with a process graph bisimilar to the process graph in Figure~\ref{fig:pdatrans}.
\end{thm}
\begin{proof}
Suppose such a guarded recursive specification does exist. Without loss of generality, we can assume it is in Greibach normal form, such that every state of its process graph is given by a sequential composition of identifiers. As in Figure~\ref{fig:pdatrans}, from every state a maximal number of consecutive $b$-steps can be executed, this is also true in the process graph of the specification, and this is also true for each of the identifiers. As the specification only has finitely many identifiers, there is a maximal number of consecutive $b$-steps any identifier can execute. Take a number $k$ larger than this maximum, and consider the state $(\uparrow, 1^k)$ of the counter, from which exactly $k$ consecutive $b$-steps can be executed. This state is bisimilar to a sequential composition $p$ of identifiers of the specification, so also from $p$, $k$ consecutive $b$-steps can be executed. By choice of $k$, these $b$-steps cannot all come from the same component of $p$, so $p$ contains two identifiers $X,Y$ ($X$ occurring before $Y$), each of which can execute at least one $b$-step. As $(\uparrow, 1^k) \downarrow$, also $p \downarrow$, and each identifier in $p$ must be accepting, so $X \downarrow$ and $Y \downarrow$. But then, the last operational rule of sequential composition can be applied to $X$, and the $b$-step(s) in $X$ can be skipped. This gives a sequence of less than $k$ consecutive $b$-steps from $p$, after which no further $b$ can be executed. But the resulting state must be bisimilar to a state $(\uparrow, 1^n)$ of the counter, with $n > 0$. From this state a further $b$ can be executed, which gives a contradiction.
\end{proof}

Transparency occurs because of the last rule of sequential composition in the operational semantics: process expression $p \seq q$ can execute a step from $q$, thereby forgetting $p$ if $p$ is accepting. When $p$ has intermediate acceptance, so $p$ can still execute a step, this step is lost.

In order to recover the correspondence between pushdown automata and sequential process algebra, we need to change this operational rule. In order not to cause confusion, we introduce a new operator $\seqc$, called \emph{sequencing}, that will replace the sequential composition operator $\seq$ of TSP. We call the resulting theory \TSPc{}. It was introduced in \cite{BLY17}.

\section{Sequential processes: \TSPc}
We replace the sequential composition $\seq$ of \TSP{} by the sequencing operator $\seqc$ of which the operational rules are shown in Figure~\ref{fig:semantics_seqc}.

\begin{figure}[htb]
  \centering
  \begin{osrules}
\osrule*{p \downarrow & q \downarrow}{p \seqc  q \downarrow}\label{osrule:seqcone}
  \qquad
    \osrule*{p\step{a} p'}{p \seqc  q \step{a} p' \seqc  q} \label{osrule:seqctwo}
  \qquad
    \osrule*{p \downarrow & p \nrightarrow & q \step{a} q' }{p \seqc  q
      \step{a} q'}\label{osrule:seqcthree}
   \end{osrules}
\caption{Operational semantics for sequencing.}
\label{fig:semantics_seqc}
\end{figure}

The last rule for sequencing has a so-called negative premise.
It is well-known that transition system specifications with negative
premises may not define a unique transition relation that agrees with
provability from the transition system specification
\cite{Gro93,BG96,Gla04}. Indeed, in \cite{BLY17} it was already
pointed out that the transition system specification in
Figure~\ref{fig:semantics_tspseq} gives rise to such anomalies in case we use unguarded recursion. Consider recursive equation
$X \defeqn X\seqc \pref[a]\emp +\emp$.  If
$X\nrightarrow$, then according to the rules for sequencing and
recursion we find that $X\step{a}\emp$, which is a contradiction. On the other hand, $X\rightarrow$ is not provable from the transition system
specification.

This is the second reason for restricting our attention to \emph{guarded}
recursive specifications.  If specification $\Delta$ is guarded,
then it is straightforward to prove that the mapping $S$ from process
expressions to natural numbers inductively defined by
$S(\emp)=S(\dl)=S(\pref[a]{p})=0$,
$S(p_1+p_2)=S(p_1\seqc p_2) = S(p_1)+S(p_2)+1$, and $S(X)=S(p)$ if
$(X\defeqn p)\in\Delta$ gives rise to a so-called
\emph{stratification} $S'$ from transitions to natural numbers defined
by $S'(p\step{a}p')=S(p)$ for all $a\in \Act \cup \{\tau\}$ and
process expressions $p$ and $p'$. In \cite{Gro93} it is proved that
whenever such a stratification exists, then the transition system
specification defines a unique transition relation that agrees with
provability in the transition system specification.

The recursive specification
  \begin{equation*}
    X \defeqn \emp + \pref[a]{Y \seqc X} \qquad
    Y \defeqn \emp + \pref[b]{\emp}+ a.Y \seqc Y \enskip.
  \end{equation*}
now yields a process graph bisimilar to the one in Figure~\ref{fig:pdatrans}, which has only binary branching. We see that it is the process of a pushdown automaton. In the next section, we will prove that every recursive specification over \TSPc{} generates a process graph that is bisimilar to the process graph of a pushdown automaton.
We find that the expressiveness of \TSPc{} is \emph{incomparable} to the expressiveness of \TSP: by means of a guarded specification over \TSPc{} we can define the always accepting counter (see Figure~\ref{fig:pda}) but not an unboundedly branching process, and for \TSP{} it is the other way around.

The operational rules presented in Fig~\ref{fig:semantics_seqc} are
in the so-called \emph{panth format} from which it immediately follows
that strong bisimilarity is a congruence \cite{Ver95}. Due to the sequencing operator, rooted branching bisimulation is no longer a congruence, and we have to add an extra condition: rooted \emph{divergence-preserving} branching bisimulation \cite{GW96, Lut20} is a congruence, which can be proved using \cite{FGL19}.

Using sequencing instead of sequential composition, the distributive
law $\A4$ is no longer valid as the following example shows.

\begin{exa}
  Consider the process expressions
    $(a.\emp + \emp)\seqc b.\emp$
  and
    $a.\emp\seqc b.\emp + \emp \seqc b.\emp$.
  On the one hand, since ${a.\emp+\emp}
  \step{a}\emp$, the last operational rule for $\seqc$ does not
  apply to $(a.\emp + \emp)\seqc b.\emp$, and hence $(a.\emp +
  \emp)\seqc b.\emp \nstep{b}$.
  On the other hand, ${\emp\seqc b.\emp}\step{b}\emp$, so ${a.\emp\seqc
     b.\emp + \emp\seqc b.\emp} \step{b} \emp$.
  It follows that
    $(a.\emp + \emp)\seqc b.\emp\not\bisim a.\emp\seqc b.\emp + \emp\seqc b.\emp$.
\end{exa}

In fact, due to this failure, there is no finite sound and ground-complete axiomatization of bisimilarity
for (the recursion-free fragment of) \TSPc{} \cite{Bel18,BLB19}.
In order to find an axiomatization, nevertheless, it suffices to add an auxiliary unary operator $\mathit{NA}$,
denoting \emph{non-acceptance}. Intuitively, $\mathit{NA}(p)$ behaves
as $p$ except that it does not have the option to accept
immediately. The operational semantics of $\mathit{NA}$ is obtained by
adding one very simple rule, see
Figure~\ref{fig:semantics_na}.

\begin{figure}[htb]
  \centering
  \begin{osrules}
     \osrule*{p \step{a} p'}{\mathit{NA}(p) \step{a} p'}
   \end{osrules}
\caption{Operational semantics for non-acceptance.}
\label{fig:semantics_na}
\end{figure}

Now, we can formulate a sound and ground-complete axiomatization of
bisimilarity for the recursion-free fragment of \TSPc{} in three parts. First, there are the axioms in Table~\ref {tab:axioms_tsp}
without the distributive law $\A4$. Second, the axiomatization of the \emph{NA} operator is
straightforward: see axioms \NA{1}--\NA{4} in Table~\ref{tab:axioms_na}.

Finally, instead of the distributive law we include axioms
\A{11}--\A{13} shown in Table~\ref{tab:axioms_na}.

\begin{table}[htb]

\centering
\begin{tabular}{lcl@{\qquad}l}
$\na{\dl}$ & $=$ & $\dl$ & $\NA1$ \\
$\na{\emp}$ & $=$ & $\dl$ & $\NA2$ \\
$\na{\pref[a]x}$ & $=$ & $\pref[a]x$ & $\NA3$ \\
$\na{x+y}$ & $=$ & $\na{x} + \na{y}$ & $\NA4$ \\
\\
$\na{x+y}\seqc z$ & $=$ & $\na{x}\seqc z + \na{y}\seqc z$ & $\A{11}$ \\
$(\pref[a]x + y + \emp)\seqc \na{z}$ & $=$ & $(\pref[a]x + y)\seqc\na{z}$ & $\A{12}$ \\
$(\pref[a]x + y + \emp)\seqc (z + \emp)$ & $=$ & $(\pref[a]x + y)\seqc (z +
                                           \emp) + \emp$ & $\A{13}$ \\
\end{tabular}
\caption{The axioms for the auxiliary operator $\mathit{NA}$ and
  weaker forms of distributivity.}
\label{tab:axioms_na}
\end{table}

The ground-completeness argument in \cite{Bel18, BLB19} proceeds via
an \emph{elimination} theorem: it is established that for every
recursion-free process expression $p$ there exists a recursion-free
process expression $q$ without occurrences of the operators $\seqc$
and $\mathit{NA}$ such that the equation $p=q$ is derivable from the
axioms above using equational logic. Thus the ground-completeness of
the axiomatisation of bisimilarity for (the recursion-free fragment)
of \TSPc{} follows from the
ground-completeness of the axiomatisation of bisimilarity for (the
recursion-free fragment of) BSP (see \cite[Theorem 4.4.12]{BBR10}).

In the context of a guarded recursive specification $\Delta$, we can,
moreover, obtain the following result.

\begin{thm}[head normal form theorem] \label{hnf}
  Let $\Delta$ be a guarded recursive specification.
  For every process expression $p$ there exists a natural number $n$
  and sequences of actions $a_1,\dots,a_n$ and process expressions
  $p_1,\dots,p_n$ such that the equation
  \begin{equation*}
     p = (\emp +{}) \sum_{i=1}^{n}a_i.p_i
  \end{equation*}
  (where $(\emp +{})$ indicates that there may or may not be a summand
  $\emp$), is derivable from the axioms for \TSPc{} and the equations in $\Delta$ using
  equational logic. The process expression $(\emp +)
  \sum_{i=1}^{n}a_i.p_i$ is called the \emph{head normal form} of $p$.
  \end{thm}
\begin{proof}
  Let $p$ be a process expression. Without loss of generality we may
  assume that $p$ is guarded. For we can replace every unguarded
  occurrence of a process identifier in $p$ by the right-hand side of
  its defining  equation in $\Delta$ which, since $\Delta$ is guarded, results in
  process expression with one fewer unguarded occurrence of a process
  identifier. We now proceed by induction on the structure of $p$.

  Note that $p$ cannot be itself a process identifier, since then $p$
  would not be guarded.

  If $p=\dl$, then $p$ is a head normal form.

  If $p=\emp$, then by \A6{} we have $p=\emp+\dl$, which
  is a head normal form.

  If $p=a.p'$, then by \A6{} and \A1{} we have that
  $p=\dl + a.p'=\sum_{i=1}^0 a_ip_i+a.p'=\sum_{i=1}^1 a_i.p_i$, with
  $a_1=a$ and $p_1=p'$, which is a head normal form.

  Suppose that $p=p'+p''$ and (using the induction hypothesis)
  that
  \begin{equation*}
    p' = (\emp +{}) \sum_{i=1}^{n'}a_{i}'.p_{i}'\ \text{and}\
    p'' = (\emp +{})\sum_{i=1}^{n''}a_{i}''.p_{i}''\enskip.
  \end{equation*}
  Then, using \A1 and  \A2 and, if
  necessary, \A3 to get rid of a superfluous occurrence
  of $\emp$, we get that
  \begin{equation*}
    p= (\emp +{}) \sum_{i=1}^{n'+n''}a_i.p_i\enskip,
  \end{equation*}
  where $a_i=a_{i}'$ and $p_i=p_{i}'$ if $1\leq i \leq n'$ and
  $a_i=a_{i-n'}''$ and $p_i = p_{i-n'}''$ if $n' < i \leq n'+n''$, and
  so $p$ has a head normal form.

  Suppose that $p=p'\seqc p''$, and (using the induction hypothesis)
  that
  \begin{equation*}
    p' = (\emp +{}) \sum_{i=1}^{n'}a_{i}'.p_{i}'\ \text{and}\
    p'' = (\emp +{})\sum_{i=1}^{n''}a_{i}''.p_{i}''\enskip.
  \end{equation*}
  We distinguish three cases:
  \begin{enumerate}
    \item Suppose that the head normal form of $p'$ does not have the
      $\emp$-summand.
      If $n'=0$, then  $p = p' \seqc p'' = \dl \seqc p'' = \dl$ and
      the latter is a head normal form.
      Otherwise, if $n'>0$, then, by axioms \NA{1}, \NA{3} and \NA{4} we have
      that $p' = \mathit{NA}(\sum_{i=1}^{n'}a_{i}'.p_{i}')$. It then
      follows with applications \A{11}, \A5{}, \NA{1}, \NA{3} and
      \NA{4} that
      \begin{equation*}
          p = p'\seqc p'' =
          \mathit{NA}(\sum_{i=1}^{n'}a_{i}'.p_{i}')\seqc p'' =
          \sum_{i=1}^{n'}\mathit{NA}(a_i'.p_i')\seqc p''
        = \sum_{i=1}^{n'}a_i'.p_i' \seqc p''\enskip.
      \end{equation*}
    \item Suppose that the head normal form of $p'$ has the
      $\emp$-summand, but the head normal form of $p''$ does not.
      If $n'=0$, then, by \A6 we have that $p'=\emp$,
      so by \A9 it follows that
      $p = p'\seqc p'' = 1\seqc p'' = p''$.
      Otherwise, if $n'>0$, then we first note that, since the head
      normal form of $p''$ does not have the $\emp$-summand, by the
      axioms for $\mathit{NA}$ we get that $p'' =
      \mathit{NA}(p'')$. Then we find with an application of the axiom
      \A{12} (and applications of \A1 and \A2 if necessary) that
      \begin{equation*}
        p = p' \seqc p'' = p' \seqc \mathit{NA}(p'') =
        \left(\sum_{i=1}^{n'}a_i'.p_{i}'\right)\seqc\mathit{NA}(p'') =
        \left(\sum_{i=1}^{n'}a_i'.p_{i}'\right)\seqc p''\enskip.
      \end{equation*}
      We can then proceed as in the first case to find the head normal
      form for $p=\left(\sum_{i=1}^{n'}a_i'.p_{i}'\right)\seqc p''$.
  \item Suppose that both the head normal forms of $p'$ and $p''$ have
    the $\emp$-summand.
    Again, if $n'=0$, then, by \A6{} we have that $p'=\emp$,
      so by \A{9} it follows that
      $p = p'\seqc p'' = \emp\seqc p'' = p''$.
    Otherwise, if $n'>0$, then we apply the axiom \A{13} to get
      $p = p' \seqc p'' = \emp +
      \left(\sum_{i=1}^{n'}a_i'.p_i'\right)\seqc p''$, we proceed
      as in the first case to find a head normal form
        $\left(\sum_{i=1}^{n'}a_i'.p_i'\right)\seqc p''$
        for $\left(\sum_{i=1}^{n'}a_i'.p_i'\right)\seqc p''$, and
        observe that
        $p = \emp + \left(\sum_{i=1}^{n'}a_i'.p_i'\right)\seqc p''$.
      \end{enumerate}

   Suppose that $p=\mathit{NA}(p')$, and (using the induction
   hypothesis) that
  \begin{equation*}
    p' = (\emp +{}) \sum_{i=1}^{n}a_{i}.p_{i}\enskip.
  \end{equation*}
  Then by the axioms for $\mathit{NA}$, \A1 and \A6 we have
  \begin{equation*}
    p=\mathit{NA}(p')=(\dl +{}) \sum_{i=1}^{n}a_i.p_i =
    \sum_{i=1}^na_i.p_i\enskip.
    \qedhere
  \end{equation*}
\end{proof}

\section{The correspondence} 

The classical theorem states that a language can be defined by a
push-down automaton just in case it can be defined by a context-free
grammar.  In our setting, we do have that the process of a given
guarded specification over \TSPc{} (i.e., the equivalence class of process graphs bisimilar to the process graph
  associated with the specification) coincides with the process of some
push-down automaton (i.e., the equivalence
  class of process graphs bisimilar to the process graph associated
  with the push-down automaton), but
not the other way around: there is a push-down automaton of which the
process is different from the process of any guarded sequential
specification. In this section, we will prove these facts; in the next
section, we investigate what is needed in addition to recover the full
correspondence.

We use the following extra notation.  If $\alpha\in\Procids^{*}$, say
$\alpha=X_1\cdots X_n$, then $\alpha$ denotes the process expression
inductively defined by $\alpha=\emp{}$ if $n=0$ and $\alpha=(X_1\cdots
X_{n-1})\seqc X_n$ if $n>0$.
We will also use this construct indexed by a word $x \in \mathcal{D}^{*}$, so if we have $X_d \in \Procids$ ($d \in \mathcal{D}$), then $\alpha_x$ is inductively defined by 
$\alpha_{\epsilon} = \emp$ and $\alpha_{xd} = \alpha_x \seqc X_d$.

First of all, we look at the failing direction. It can fail if the push-down automaton has at least two states. For one state, it does work.

\begin{thm} \label{th:onestate}
For every one-state pushdown automaton there is a guarded sequential
specification of which the process coincides with the process of the automaton.
\end{thm}
\begin{proof}
Let $M = (\{\uparrow\},\mathcal{A},\mathcal{D},{\rightarrow},{\uparrow},{\downarrow})$ be a pushdown automaton. We have identifiers $X$ and $X_d$ for $d \in \mathcal{D}$.
If there is no transition 
${\uparrow} \xrightarrow{a[\epsilon/x]} {\uparrow}$, then we can take either $X = \emp$ or $X = \dl$ as the resulting specification (in case ${\downarrow} = \{\uparrow\}$ resp. ${\downarrow} = \emptyset$). Otherwise, add a summand $a.\alpha_x \seqc X$ for each such transition to the equation of the initial identifier $X$.
Next, the equation for the added identifier $X_d$ has a summand $a.\alpha_x$ for each transition ${\uparrow} \xrightarrow{a[d/x]} {\uparrow}$, and a summand $\emp$ or $\dl$, depending on whether ${\uparrow} \in {\downarrow}$ or not.

Now a bisimulation between the process graph of $M$ and the process graph of the constructed specification can be obtained by relating a state $(\uparrow, x)$ ($x \in \mathcal{D}^{*}$) to the state given by the sequence of identifiers $\alpha_x \seqc X$. The initial states are related (by identifying $X$ and $\emp \seqc X$) and $(\uparrow, dx) \step{a} (\uparrow, yx)$ just in case $X_d \seqc \alpha_x \seqc X \step{a} \alpha_y \seqc \alpha_x \seqc X$. Also, $(\uparrow, \epsilon) \step{a} (\uparrow, x)$ just in case $X \step{a} \alpha_x \seqc X$.
\end{proof}

\begin{exa} \label{ex:onestate}
The counter of Figure~\ref{fig:pdatrans} can be seen as a stack over a singleton data set $\{1\}$. For a general finite data set $\mathcal{D}$,
the stack that is accepting in every state has a pushdown automaton with state $\uparrow$, actions $\{ push_d, pop_d \mid d \in \mathcal{D}\}$, ${\downarrow} = \{\uparrow\}$ and 
transitions ${\uparrow} \xrightarrow{push_d[\epsilon/d]} {\uparrow}$ and ${\uparrow} \xrightarrow{pop_d[d/\epsilon]} {\uparrow}$ for each $d \in \mathcal{D}$ and transitions ${\uparrow} \xrightarrow{push_e[d/ed]} {\uparrow}$ for each $d,e \in \mathcal{D}$.

The following guarded recursive specification is obtained for this pushdown automaton in accordance with the procedure described in the proof of Theorem~\ref{th:onestate}.
 \begin{equation*}
    X \defeqn \emp + \sum_{d \in \mathcal{D}} push_d.X_d \seqc X \qquad
    X_d \defeqn \emp + pop_d.\emp + \sum_{e \in D} push_e.X_e \seqc X_d \enskip (d \in \mathcal{D})
  \end{equation*}
\end{exa}

Note that if we use the sequential composition operator $\seq$ of \TSP{} \cite{BBR10} instead of the present sequencing operator $\seqc$ of \TSPc, then Theorem~\ref{th:onestate} fails because of the transparency illustrated in Figure~\ref{fig:UnboundedBranching}. With the sequential composition operator, we cannot find a recursive specification of the stack accepting in every state of Example~\ref{ex:onestate}, because of the same phenomenon. 

\begin{figure}[htb]
\begin{center}
\begin{tikzpicture}[->,>=stealth',node distance=2cm, node font=\footnotesize, state/.style={circle, draw, minimum size=.5cm,inner sep=0pt}]
  \node[state,initial,initial text={},initial where=above] (s0) {$\uparrow$};
  \node[state,accepting] [right of=s0] (s1) {$\downarrow$};
  
  \path[->]
  (s0) edge[in=225,out=135,loop]
         node[left] {$\begin{array}{c}a[\epsilon/1]\\
                                      a[1/11]\\
                                      b[1/\epsilon]\end{array}$} (s0)
  (s0) edge node[above] {$\begin{array}{c}c[\epsilon/\epsilon]\\ c[1/1]\end{array}$} (s1)
  (s1) edge[in=-45,out=45,loop] node[right] {$b[1/\epsilon]$} (s1);
\end{tikzpicture}
\end{center}
\caption{Pushdown automaton used in the proof of Theorem~\ref{nospec}.}\label{fig:pda2}
\end{figure}

\begin{figure}[htb]
\begin{center}
\begin{tikzpicture}[->,>=stealth',node distance=2cm, node
  font=\footnotesize, state/.style={ellipse, draw, minimum size=.5cm,inner sep=0pt}]
  \node[state,initial,initial text={},initial where=left] (s0)
  {$(\uparrow,\epsilon)$};
  \node[state] [right of=s0] (s1) {$(\uparrow,1)$};
  \node[state] [right of=s1] (s2) {$(\uparrow,11)$};
  \node[state, draw=none] [right of=s2] (sdots) {$\dots$};
  \node[state,accepting] [below of=s0] (t0) {$(\downarrow,\epsilon)$};
  \node[state,accepting] [right of=t0] (t1) {$(\downarrow,1)$};
  \node[state,accepting] [right of=t1] (t2) {$(\downarrow,11)$};
  \node[state, draw=none] [right of=t2] (tdots) {$\dots$};

  \path[->]
    (s0) edge[bend left] node[above] {$a$} (s1)
    (s1) edge[bend left] node[above] {$a$} (s2)
    (s2) edge[bend left] node[above] {$a$} (sdots);
  \path[->]
    (sdots) edge[bend left] node[below] {$b$} (s2)
    (s2) edge[bend left] node[below] {$b$} (s1)
    (s1) edge[bend left] node[below] {$b$} (s0);
  \path[->]
    (tdots) edge node[below] {$b$} (t2)
    (t2) edge node[below] {$b$} (t1)
    (t1) edge node[below] {$b$} (t0);
  \path[->]
    (s0) edge node[right] {$c$} (t0)
    (s1) edge node[right] {$c$} (t1)
    (s2) edge node[right] {$c$} (t2);
\end{tikzpicture}
\end{center}
\caption{The process graph associated with the pushdown automaton
  in Figure~\ref{fig:pda2}.}\label{fig:pdatrans2}
\end{figure}

\begin{thm} \label{nospec}
There is a pushdown automaton with two states, such that there is no
guarded sequential specification with the same process.
\end{thm}
\begin{proof}
Consider the example pushdown automaton in Figure~\ref{fig:pda2}.
Suppose there is a finite guarded sequential specification with the
same process, depicted by the representative in
Figure~\ref{fig:pdatrans2}. Without loss of generality we can assume
that this specification is in Greibach Normal Form (see
\cite{BLB19}). As a consequence, each state of the process graph generated by the automaton is bisimilar to a sequence of identifiers of the specification (as defined earlier).
Take $k$ a natural number that is larger than the number of process identifiers of the specification, and for $ 0 \leq i \leq k$ consider the state $(\uparrow,1^i)$ reached after executing $i$ $a$-steps. From this state, consider any sequence of steps $\steps{w}$ where $a \not\in w$. Thus, $w$ contains at most one $c$ and at most $i$ $b$'s. 

In the process graph generated by the recursive specification, this same sequence of steps $\: {\steps{w}}$ is possible from the sequence of identifiers $\alpha_i$ bisimilar to state $(\uparrow,1^i)$. Let $X_i$ be the first element of $\alpha_i$. From $X_i$, we can also execute at most one $c$-step and $i$ $b$-steps, without executing an $a$-step.

Since $k$ is larger than the number of process identifiers of the specification, there must be a repetition in the identifiers $X_i$ ($i \leq k$). Thus, there are numbers $n,m$, $n<m \leq k$, with $X_n = X_m$. The process identifier $X_m$ can execute at most one $c$ and $n < m$ $b$'s without executing an $a$.
But $\alpha_m\steps{b^m}$, so the additional $b$-steps must come from the second and following identifiers of the sequence. As the second identifier is reached by just executing $b$'s, this is a state reached by just executing $a$'s and $b$'s, so it must allow an initial $c$-step.
Now we can consider $\alpha_m\steps{cb^m}$. This sequence of steps must also reach the second identifier, but then, a second $c$ can be executed, contradicting that $\alpha_i$ is bisimilar to $(\uparrow,1^i)$.

Thus, our assumption was wrong, and the theorem is proved.
\end{proof}

We see that the contradiction is reached, because when we reach the second identifier in the sequence $\alpha_i$, we do not know whether we are in a state relating to the initial state or the final state of the pushdown automaton. Going from the first identifier to the second identifier by means of the sequencing operator, no extra information can be passed along. In the next section we will add a mechanism that allows the passing of extra information along the sequencing operator.

Theorem~\ref{nospec} holds for the sequencing operator we introduced, but it holds in the same way for the sequential composition operator of \TSP{} \cite{BBR10}. No intermediate acceptance is involved in the proof.

In the other direction, we can find a pushdown automaton with the same
process as a given guarded sequential specification. The proof will be
more complicated than the classical proof, where it is only needed to
find a pushdown automaton with the same language. The classical proof
uses the equivalence of acceptance by final state and acceptance by
empty stack, which does not hold in bisimulation semantics. 

We again have to deal with the failure of the distributive law. Due to this,
we carefully need to consider every instance of intermediate
acceptance. Consider the sequencing $(a.\emp + \emp) \seqc
b.\emp$. The left argument of this sequencing shows intermediate
acceptance, the right argument does not. As $(a.\emp + \emp) \seqc b.\emp \bisim a.\emp \seqc b.\emp$, the intermediate acceptance in the first argument is redundant, and can be removed. We have to restrict the notion of Greibach normal form, in order to remove all redundant intermediate acceptance.

We proceed to recall the definition of Acceptance Irredundant Greibach
Normal Form from \cite{BLB19}.
First, partition the set of process identifiers $\Procids$ into sets
  $\ProcidsT=\{X\in\Procids\mid \term{X}\}$
and
  $\ProcidsNT=\{X\in\Procids\mid\nterm{X}\}$.
Then, define the set $\ProcidsINT$ of \emph{hereditarily
non-accepting} process identifiers as the largest subset of 
$\ProcidsNT$ such that for all $X\in\ProcidsINT$ we have that
if $(X\defeqn p)\in\Delta$ and $Y$ is a process identifier occurring
in $p$, then $Y\in\ProcidsINT$.
We say that $\alpha\in\Procids^{*}$ is \emph{acceptance
  irredundant}
if $\alpha\in\ProcidsINT^{*}\ProcidsNT\ProcidsT^{*}\cup\ProcidsT^{*}$.
  
  \begin{defi}\label{def:AIGNF}
    A recursive specification $\Delta$ is in \emph{Acceptance Irredundant Greibach
      Normal Form} (AIGNF) if every right-hand side of every equation
    has the following form:
    \begin{itemize}
\item $(\emp +) \sum_{i=1}^{n} a_i.\alpha_i$ for actions $a_i \in
  \mathcal{A} \cup \{\tau\}$, acceptance irredundant sequences of identifiers $\alpha_i
  \in \mathcal{P}^{*}$, and
  $n \geq 0$.
\end{itemize}
   \end{defi}

   According to \cite[Proposition 20]{BLB19}, for every recursive
   specification $\Delta$ there exist $\Procids'\supseteq\Procids$ and
   a recursive specification $\Delta'$ in AIGNF such that for all
   $X,Y\in\Procids$ we have that $X\bisim Y$ with respect to $\Delta$
   if, and only if, $X\bisim Y$ with respect to $\Delta'$. Moreover,
   every state of the process graph associated with a process
   identifier is given by an acceptance irredundant sequence of identifiers.

   Now that we can assume guarded sequential specifications to be in
   AIGNF, we are one step closer to associating a pushdown automaton
   with each of them. There is, however, still one complication we
   need to resolve, as illustrated by the following example.

   \begin{exa} \label{exa:nonseparation}
  Consider the AIGNF specification
  \begin{equation*}
    X \defeqn \emp + a.Y \seqc X\qquad
    Y \defeqn b.\emp + a.Y \seqc Y\enskip.
  \end{equation*}
  The sequence $YYX$ reachable from $X$ is non-accepting, and since
  $Y\step{b}\emp$, we have that $YYX\step{b} YX$. The sequence $YX$ is
  non-accepting too, but we have that $YX\step{b}X$ and $X$ is
  accepting. Thus, we see that the $b$-transition emanating from $Y$ 
  can lead to a non-accepting state or to an accepting state,
  depending on the context. So it cannot be determined from
  the first process identifier of a sequence of process identifiers alone
  whether executing a transition from it will lead to an accepting state or a
  non-accepting state.
\end{exa}

  To determine when a pushdown automaton executing a sequential
  process needs to switch from non-accepting to accepting, it would
  be convenient if, given any acceptance irredundant sequence
  $\alpha\beta$ with $\alpha$ a sequence of non-accepting process
  identifiers and $\beta$ a sequence of accepting process identifiers:
  \begin{itemize}
    \item the rightmost identifier in $\alpha$ is always from a
      designated subset $\Procidssep$ of non-accepting
      identifiers;
   \item all the other process identifiers in $\alpha$ are not in
     $\Procidssep$ (nor can they reach a process identifier in $\Procidssep$).
  \end{itemize}
  Let $\Procidssep\subseteq\ProcidsNT$; we say $\Procidssep$
  \emph{separates non-acceptance from acceptance} in $\alpha$ if
  \begin{equation*}
    \alpha\in
    (\ProcidsINT-\Procidssep)^{*}\Procidssep\ProcidsT^{*}\cup\ProcidsT^{*}\enskip.
  \end{equation*}
  If there exists a set $\Procidssep\subseteq\ProcidsNT$
  such that $\Procidssep$ separates non-acceptance from acceptance in
  all sequences of process identifiers associated with the states of a
  process graph, then a pushdown automaton simulating the behaviour of
  that process graph can decide to switch from non-acceptance to acceptance
  if, and only if, it executes a transition from some process
  identifier in $\Procidssep$ that leads to an accepting sequence
  of process identifiers.

  We now define when a subset of non-accepting process identifiers
  separates non-acceptance from acceptance in a given sequential
  specification in AIGNF.

  \begin{defi}\label{def:restrAIGNF}
    Let $\Delta$ be in AIGNF and let
    $\Procidssep\subseteq\ProcidsNT$. We say that $\Procidssep$
    \emph{separates non-acceptance from acceptance}  in $\Delta$ if
    for all sequences of process identifiers $\alpha$ and $\beta$ such
    that $\alpha\step{a}\beta$ for some $a\in \mathcal{A} \cup
    \{\tau\}$ if holds that whenever $\Procidssep$ separates
    non-acceptance from acceptance in $\alpha$, then $\Procidssep$
    separates non-acceptance from acceptance in $\beta$.
  \end{defi}

  Note that for the AIGNF specification of
  Example~\ref{exa:nonseparation} there does not exist a subset of
  non-accepting process identifiers that separates non-acceptance from
  acceptance, for we have that $YYX$ is reachable from $X$ and in
  $YYX$ non-acceptance cannot be separated from acceptance. We shall
  illustrate below how, by adding an extra process identifier, the
  recursive specification of Example~\ref{exa:nonseparation} can be
  transformed into a recursive specification in which $\{Y\}$
  separates non-acceptance from acceptance.

    \begin{exa}\label{exa:separate}
      The specification of Example~\ref{exa:nonseparation} can be
      transformed to a specification in AIGNF in which $\{Y\}$
      separates non-acceptance from acceptance by introducing a new
      process identifier $Z$ and changing the specification to
  \begin{equation*}
    X \defeqn \emp + a.Y \seqc X\qquad
    Y \defeqn b.\emp + a.Z \seqc Y\qquad
    Z \defeqn b.\emp + a.Z \seqc Z\enskip.
  \end{equation*}
    We argue that $\{Y\}$ separates non-acceptance from acceptance in
    this recursive specification.
    To this end, suppose that $\alpha\step{x}\beta$
    with $x\in\{a,b\}$. If $\{Y\}$ separates non-acceptance from
    acceptance in $\alpha$, then either $\alpha\in\{X\}^{*}$, or
    $\alpha=Y\alpha'$ and $\alpha'\in\{X\}^{*}$, or
    $\alpha=Z\alpha'$ and $\alpha'\in\{Z\}^{*}\{Y\}\{X\}^{*}$.

    If $\alpha\in\{X\}^{*}$, then $x=a$ and $\beta=Y\alpha\in\{Z\}^{*}\{Y\}\{X\}^{*}$.

    If $\alpha=Y\alpha'$ and $\alpha'\in\{X\}^{*}$, then either $x=b$ and
    $\beta=\alpha'\in\{X\}^{*}$ or $x=a$ and $\beta={{ZY}\alpha'}\in{\{Z\}^{*}\{Y\}\{X\}^{*}}$.

    If $\alpha=Z\alpha'$ and  $\alpha'\in\{Z\}^{*}\{Y\}\{X\}^{*}$,  then either $x=b$
    and $\beta=\alpha'\in\{Z\}^{*}\{Y\}\{X\}^{*}$ or $x=a$ and
    $\beta={Z\alpha}\in{\{Z\}^{*}\{Y\}\{X\}^{*}}$.
  \end{exa}

    We generalise the insight of the preceding example in the
    following proposition.

    \begin{prop} \label{prop:separation}
      For every recursive specification $\Delta$ defining the process
      identifiers in $\Procids$ there exist a set of process
      identifiers $\Procids'\supseteq\Procids$ and a recursive
      specification $\Delta'$ defining the process identifiers in
      $\Procids'$ such that $\Delta'$ is in AIGNF and $\ProcidsNT$
      separates non-acceptance from acceptance in $\Delta'$.
      Moreover, for all $X,Y\in\Procids$ we have that $X\bisim Y$ with
      respect to $\Delta$ if, and only if, $X\bisim Y$ with respect to $\Delta'$.
    \end{prop}
    \begin{proof}
      By \cite[Proposition 20]{BLB19} we can assume without loss of
      generality that $\Delta$ is in AIGNF.
      Let $\ProcidsNT^{\dagger}$ be a set of process identifiers disjoint from
      $\Procids$ and with the same cardinality as $\ProcidsNT$; with
      every $X\in\ProcidsNT$ we bijectively associate a process identifier in
      $\ProcidsNT^{\dagger}$ that we denote by $X^{\dagger}$. Furthermore, if
      $\alpha\in\Procids^{*}$ is an acceptance irredundant sequence of
      identifiers, then we denote by $\alpha^{\dagger}$ the sequence of
      identifiers in $\Procids'=(\Procids\cup\ProcidsNT^{\dagger})^{*}$ that is obtained
      from $\alpha$ by replacing all occurrences of non-accepting
      process identifiers in $\alpha$ except the last one (if it
      exists) by their variants. That is, for all
      acceptance-irredundant sequences of process identifiers
      $\alpha$, we define $\alpha^{\dagger}$ with induction on the length of $\alpha$ as follows:
      \begin{itemize}
      \item if $\alpha\in\ProcidsT^{*}$ or
        $\alpha\in\ProcidsNT\ProcidsT^{*}$,  then
        $\alpha^{\dagger}=\alpha$; and
      \item if $\alpha=X\alpha'$ and $\alpha'\in\ProcidsINT^{*}\ProcidsNT\ProcidsT^{*}$, then
        $\alpha^{\dagger}=X^{\dagger}(\alpha')^{\dagger}$.
      \end{itemize}

      Define the recursive specification $\Delta'$ for all process
      identifiers $Y\in\Procids'$ as follows:
      \begin{enumerate}
      \item if $Y\in\Procids$ and $\Delta(Y)=(\emp{}+)\sum_{i=1}^na_i.\alpha_i$, then
        $\Delta^{\dagger}(Y)=(\emp{}+)\sum_{i=1}^na_i.\alpha_i^{\dagger}$;
      \item if $Y=X^{\dagger}$ for some $X\in\ProcidsNT$ and
        $\Delta(X)=(\emp{}+)\sum_{i=1}^na_i.\alpha_i$, then
        $\Delta^{\dagger}(Y)=(\emp{}+)\sum_{i=1}^na_i.\alpha_i^{\ddagger}$,
        where $\alpha_i^{\ddagger}$ is obtained from $\alpha$ by
        replacing \emph{all} occurrences (including the last) of non-accepting process
        identifiers $Z$ in $\alpha$ by their variant $Z^{\dagger}$.
      \end{enumerate}

      We first argue that for all $X, Y\in\Procids$ we have that
      $X\bisim Y$ with respect to $\Delta$ if, and only if, $X\bisim
      Y$ with respect to $\Delta'$. Let
      $f:(\Procids')^{*}\rightarrow\Procids^{*}$ be the
      mapping that maps every sequence $\alpha\in (\Procids')^{*}$
      to the sequence obtained by replacing in $\alpha$ every occurrence of a process identifier
      $X^{\dagger}\in\ProcidsNT^{\dagger}$ by $X$. It is immediate from the
      definition of $\Delta'$ that for all $\alpha\in  (\Procids')^{*}$ we have that
      \begin{itemize}
        \item if $\alpha\step{a}\beta$ with respect to $\Delta'$
      then
      $f(\alpha)\step{a}f(\beta)$ with respect to $\Delta$, and
        \item if $f(\alpha)\step{a}\alpha'$ with respect to $\Delta$, then there exists
       $\beta\in (\Procids')^{*}$ such that  $\alpha\step{a}\beta$
       with respect to $\Delta'$ and  $f(\beta)=\alpha'$.
     \end{itemize}
      It immediately follows from these two properties that if $R$ is a
      bisimulation relation with respect to $\Delta$, then $\{(\alpha,\beta)\in (\Procids')^{*}\times
      (\Procids')^{*}\mid (f(\alpha),f(\beta))\in R\}$ is a
      bisimulation with respect to $\Delta'$, and if $R$ is a
      bisimulation relation with respect to $\Delta'$, then
      $\{(f(\alpha),f(\beta))\in\Procids^{*}\times\Procids^{*}\mid
      (\alpha,\beta)\in R\}$ is a bisimulation with respect to $\Delta$. Hence,
      $X\bisim Y$ with respect to $\Delta$ if, and only if, $X\bisim
      Y$ with respect to $\Delta'$.
      
      It remains to argue that if $\alpha$ is an acceptance
      irredundant sequence such that $\ProcidsNT$ separates non-acceptance from acceptance in
      $\alpha$ and
      $\alpha\step{a}\beta$ for some action $a$, then also $\beta$ is an acceptance
      irredundant sequence and $\ProcidsNT$ separates non-acceptance from
      acceptance in $\beta$.
      
      If $\alpha$ is the empty sequence, then
      $\alpha\nstep{a}$, so there is nothing to prove.

      If $\alpha$ is not the empty sequence, then there is a process
      identifier $Y$ such that $\alpha = Y\alpha'$.
      We assume that $\alpha\step{a}\beta$ and we distinguish
      cases according to whether $Y\in \Procids$ or $Y\in\ProcidsNT^{\dagger}$:
      \begin{enumerate}
      \item If $Y\in\Procids$, then
          $\Delta'(Y)=(\emp{}+)\sum_{i=1}^na_i.\alpha_i^{\dagger}$
        and $\beta=\alpha_i^{\dagger}\alpha'$ for some $1\leq i \leq n$. Since $Y\in\Procids$
        and $\ProcidsNT$ separates non-acceptance from acceptance in $\alpha$, it follows that
        $\alpha'\in\ProcidsT^{*}$. Hence, since $\alpha_i^{\dagger}$ is
        acceptance irredundant and $\ProcidsNT$ separates non-acceptance from
        acceptance in $\alpha_i^{\dagger}$ it follows that $\beta$ is also acceptance
        irredundant and $\ProcidsNT$ separates non-acceptance from
        acceptance in $\beta$.
      \item If $Y=X^{\dagger}$ for some $X\in\ProcidsNT$ and
        $\Delta(X)=(\emp{}+)\sum_{i=1}^na_i.\alpha_i$,
        then $\beta=\alpha_i^{\ddagger}\alpha'$ for some $1\leq i \leq
        n$. Since $\ProcidsNT$ separates non-acceptance from
        acceptance in $\alpha$ and $Y\in\ProcidsNT^{\dagger}$, it follows that $\alpha'\in
        (\ProcidsNT^{\dagger})^{*}\ProcidsNT\ProcidsT^{*}$. 
        Moreover, since $\Delta$ is in AIGNF and $f(\alpha)$ is
          acceptance-irredundant, it follows that $f(\beta)$ and
          hence $\beta$ is acceptance-irredundant. It follows that
          $f(\alpha_i^{\ddagger})$ is either the empty sequence or a
          sequence of hereditarily non-accepting process identifiers, and therefore $\alpha_i^{\ddagger}\in(\ProcidsNT^{\dagger})^{*}$. Thus, we
          find that
          $\beta=\alpha_i^{\ddagger}\alpha'\in(\ProcidsNT^{\dagger})^{*}\ProcidsNT\ProcidsT^{*}$,
          so $\ProcidsNT$ separates non-acceptance from acceptance in $\beta$.
          \qedhere
      \end{enumerate}
    \end{proof}
         
    We can now associate with every guarded sequential specification a
    suitable pushdown automaton.
\begin{thm} \label{theo3}
For every guarded sequential specification there is a  pushdown
automaton with a bisimilar process graph, with at most two states.
\end{thm}
\begin{proof}
Let $\Delta$ be a guarded sequential specification over $\Procids$. By
Proposition~\ref{prop:separation} we can assume without loss of
generality that $\Delta$ is in Acceptance Irredundant Greibach Normal
Form and separates non-acceptance from acceptance. Moreover, there
exists a subset $\Procidssep$ of non-accepting process identifiers
such that every state of the specification is given by a sequence of
identifiers that is acceptance irredundant and $\Procidssep$ separates
non-acceptance from acceptance. The corresponding pushdown automaton has two states $\{n,t\}$. The initial state is $n$ iff the initial identifier $S\not\downarrow$ and $t$ iff the initial identifier $S\downarrow$ (as defined by the operational semantics), and the final state is $t$.
\begin{itemize}
\item For each summand $a.\alpha$ of an identifier $X$ with $X\downarrow$ and $\alpha = \emp$ or the first identifier of $\alpha$ is an identifier with $\downarrow$, add a step $t \xrightarrow{a[X/\alpha]} t$. Moreover, in case $X$ is initial,  a step $t \xrightarrow{a[\epsilon/\alpha]} t$;
\item For each summand $a.\alpha$ of an identifier $X$ with $X\downarrow$ and the first identifier of $\alpha$ an identifier with $\not\downarrow$, add a step $t \xrightarrow{a[X/\alpha]} n$. Moreover, in case $X$ is initial,  a step $t \xrightarrow{a[\epsilon/\alpha]} n$;
\item For each summand $a.\alpha$ of an identifier $X\in\Procidssep$ and $\alpha = \emp$ or the first identifier of $\alpha$ is an identifier with $\downarrow$, add a step $n \xrightarrow{a[X/\alpha]} t$. Moreover, in case $X$ is initial,  a step $n \xrightarrow{a[\epsilon/\alpha]} t$;
\item For each summand $a.\alpha$ of an identifier $X\in\ProcidsNT-\Procidssep$, add a step $n \xrightarrow{a[X/\alpha]} n$. Moreover, in case $X$ is initial,  a step $n \xrightarrow{a[\epsilon/\alpha]} n$.
\end{itemize}
Now the bisimulation relation will relate $\emp$ to $(t, \epsilon)$ and will relate $X_d \seqc \alpha_x$ to $(t,x)$ if $X_d \downarrow$ and to $(n,x)$ if $X_d \not\downarrow$ for each $x \in \mathcal{P}^{*}$.
Now it is not difficult to check that the process of this pushdown
automaton is the same as the process of the given guarded sequential specification.
\end{proof}

\begin{figure}[htb]
\begin{center}
\begin{tikzpicture}[->,>=stealth',node distance=2cm, node font=\footnotesize, state/.style={circle, draw, minimum size=.5cm,inner sep=0pt}]
  \node[state,initial,initial text={},initial where=above] (s0) {$n$};
  \node[state,accepting] [right of=s0] (s1) {$t$};
  
  \path[->]
  (s0) edge[in=225,out=135,loop]
         node[left] {$\begin{array}{c}a[\epsilon/X]\\
                                      a[X/XY]\end{array}$} (s0)
  (s0) edge node[above] {$b[X/\epsilon]$} (s1)
  (s1) edge[in=-45,out=45,loop] node[right] {$c[Y/\epsilon]$} (s1);
\end{tikzpicture}
\end{center}
\caption{Pushdown automaton illustrating the construction in the proof
  of Theorem~\ref{theo3}.}\label{fig:pda3}
\end{figure}

\begin{exa}
Consider the recursive specification
 \begin{equation*}
    X \defeqn a.X \seqc Y + b.\emp \qquad
    Y \defeqn \emp + c.\emp
  \end{equation*}
Notice it is in Acceptance Irredundant GNF. Notice that for all states in the process graph, there is never a sequence of more than one non-accepting identifier.
We obtain the pushdown automaton shown in Figure~\ref{fig:pda3}.
\end{exa}

In the case of the sequential composition operator of \TSP{} \cite{BBR10}, Theorem~\ref{theo3} fails because a guarded recursive specification over \TSP{} can generate a process graph with unbounded branching, that cannot be bisimilar to a process graph generated by a pushdown automaton (see Example~\ref{exa:difference}).

Note that Theorem~\ref{theo3} also fails when we use pushdown automata with acceptance by empty stack, for by means of a guarded specification over \TSPc{} (or \TSP) we can generate a process graph with infinitely many non-bisimilar accepting states.

The conclusion of this section is, that the replacement of sequential composition $\seq$ by sequencing $\seqc$ results in a calculus that is not \emph{too} expressive: every guarded specification yields the process of a pushdown automaton. On the other hand, this calculus is not expressive \emph{enough}: we still cannot specify all pushdown processes. In order to do that, we need an extra ingredient.

\section{Signals and conditions}

In order to obtain the missing correspondence, we need a mechanism to
pass state information along a sequencing operator. We shall prove
that it suffices to add the mechanism provided by propositional
signals together with a conditional statement as given in \cite{BB97},
see also \cite{BBR10}.

To keep the focus of attention on the process calculus, especially
also when we consider a sound and ground-complete axiomatization for it, we
want to concern ourselves as little as possible with the formalities
of propositional logic. To this end, we presuppose a Boolean algebra
$\mathcal{B}$, with distinguished elements $\mathit{false}$ and
$\mathit{true}$, a unary operator $\neg$ and binary operators $\vee$
and $\wedge$, freely generated by some suitable (finite or countably
infinite) set of generators. The elements of this Boolean algebra can
be thought of as abstract propositions modulo logical equivalence, and each
generator as the logical equivalence class of some propositional
variable. A \emph{valuation} is a homomorphism from $\mathcal{B}$
into the two-element Boolean algebra; it is completely determined by
how it maps the generators. Moreover, we have that if $\phi$ and
$\psi$ are elements of $\mathcal{B}$ and $v(\phi)=v(\psi)$ for all
valuations $v$, then $\phi=\psi$. Henceforth, for the sake of
readability, we will commit a minor abuse of language by referring to the
elements of $\mathcal{B}$ as \emph{propositions} and to the generators
as \emph{propositional  variables}.

Given an element $\phi$ of $\mathcal{B}$ and a process expression $p$, we
write $\phi : \rightarrow p$, with the intuitive meaning `\emph{if}
$\phi$ \emph{then} $p$'; the construct is referred to as
\emph{conditional} or \emph{guarded command}\index{guarded command}.
The operational behaviour of $\phi:\rightarrow p$ depends on a
valuation that associates a truth value with $\phi$.
Upon executing an action $a$ in a state with valuation $v$,
a state with a possibly different valuation $v'$ results. The
resulting valuation $v'$ is called the \emph{effect} of the execution
of $a$ in a state with valuation $v$.

We present operational rules for guarded command in
Figure~\ref{GCsostable}; it presupposes a function $\mathit{effect}$ that associates with every action $a$ and
  every valuation $v$ its effect. We define when a process expression
  together with a certain valuation can take a step or be in a final state.

 \begin{figure}[htb]
 \centering
 
 \begin{osrules}
 \osrule*{}{\langle \emp,v \rangle \downarrow} \qquad \osrule*{v' = \mathit{effect}(a,v)}{\langle a.p,v \rangle \step{a} \langle p,v' \rangle} 
 \qquad     \osrule*{\langle p,v \rangle \step{a} \langle p',v' \rangle}{\langle \mathit{NA}(p),v \rangle \step{a} \langle p',v' \rangle} \\
   \osrule*{\langle p,v \rangle \step{a} \langle p',v' \rangle}{\langle p + q,v \rangle \step{a} \langle p',v'\rangle \;\; \langle q + p,v \rangle \step{a} \langle p',v'\rangle} \quad
  \osrule*{\langle p,v \rangle \downarrow}{\langle p + q,v \rangle \downarrow \;\; \langle q+p,v \rangle \downarrow} \\
     \osrule*{\langle p,v \rangle \downarrow & \langle q,v \rangle \downarrow}{\langle p \seqc q,v \rangle \downarrow}
  \quad
    \osrule*{\langle p,v \rangle \step{a} \langle p',v' \rangle}{\langle p \seqc q,v \rangle \step{a} \langle p' \seqc q,v' \rangle}
    \quad
  \\
    \osrule*{\langle p,v \rangle \downarrow & \langle p,v \rangle \nrightarrow & \langle q,v \rangle \step{a} \langle q',v' \rangle}{\langle p \seqc q,v \rangle
      \step{a} \langle q',v' \rangle}
    \\
    \osrule*{\langle p,v \rangle \step{a} \langle p',v' \rangle & X\defeqn p}{\langle X,v \rangle \step{a} \langle p',v' \rangle}
  \qquad
  \osrule*{\term{\langle p,v \rangle} & X\defeqn p}{\term{\langle X,v \rangle}} \\
   \osrule*{\langle p,v \rangle \step{a} \langle p',v' \rangle \quad v(\phi) = \mathit{true}}{\langle \phi : \rightarrow p,v \rangle \step{a} \langle p',v' \rangle} \qquad
    \osrule*{\langle p,v \rangle \downarrow  \quad v(\phi) = \mathit{true}}{\langle \phi : \rightarrow p,v \rangle \downarrow} 
 \end{osrules}
 \caption{Operational rules for guarded command ($a \in \mathcal{A}
   \cup \{\tau\}$ and $\phi$ ranging over propositions).}
 \label{GCsostable} %
\end{figure}

On the basis of these rules, we can define a notion of
bisimulation. We use \emph{stateless} bisimulation, which means that
two process graphs are bisimilar iff there is a bisimulation relation
that relates two process expressions iff they are related under every
possible valuation. The stateless bisimulation also allows
non-determinism, as the effect of the execution of an action will
allow every possible resulting sequel. See the example further on,
after we also introduce the root signal operator.

\begin{defi} \label{def:statelessbisimilarity}
A binary relation $R$ on the set of sequential process expressions with conditionals is a \emph{stateless bisimulation} iff $R$ is
symmetric and for all process expressions $p$ and $q$ such that if $(p,q) \in R$:
\begin{enumerate}
\item If for some $a \in \mathcal{A} \cup \{\tau\}$ and valuations
  $v,v'$ we have $\langle p,v \rangle \step{a} \langle p',v' \rangle$,
  then there exists a process expression $q'$, such that $\langle q,v \rangle \step{a} \langle q',v' \rangle$, and $(p', q') \in R$.
\item If for some valuation $v$ we have $\term{\langle p,v \rangle}$, then $\term{\langle q,v \rangle}$.
\end{enumerate}
The process expressions $p$ and $q$ are stateless bisimilar (notation: $p \sbisim q$) iff there exists a stateless bisimulation $R$ such that $(p, q) \in R$.
\end{defi}

Stateless bisimulation is a congruence for sequential process
expressions with conditionals (proven along the lines of
\cite{BBR10}), and so we can investigate the equational theory. We add
the axioms in Table~\ref{tab:axioms_cnd}.

\begin{table}[htb]
\centering
\begin{tabular}{lcl@{\quad}l}
  $\mathit{true} : \rightarrow x$ & $=$ & $x$ & \C{1} \\
  $\mathit{false} : \rightarrow x$ & $=$ & $\dl$ & \C{2}\\
  $(\phi \vee \psi) : \rightarrow x$ & $=$ & $(\phi : \rightarrow x) + (\psi : \rightarrow x)$ &
  \C{3} \\
  $(\phi \wedge \psi) : \rightarrow x$ & $=$ & $\phi : \rightarrow (\psi : \rightarrow x)$ &
  \C{4} \\
  $\phi : \rightarrow (x+y)$  & $=$ & $(\phi : \rightarrow x) + (\phi
                                      : \rightarrow y)$ & \C{5} \\
  $\phi : \rightarrow (x\seqc y)$ & $=$ & $(\phi : \rightarrow x)\seqc
                                        y$ &
  \C{6} \\ \\
  $\phi : \rightarrow \na{x}$ & $=$ & $\na{\phi : \rightarrow x}$ &
                                                                    \C{7}\\
  \multicolumn{4}{l}{$(\na{x}+\phi: \rightarrow \emp);(\na{y}+\psi: \rightarrow \emp{})$}\\
                                  & $=$ &
  $\na{x};(\na{y}+\psi:\rightarrow \emp) + (\phi\wedge\psi) :
  \rightarrow \emp{}$ & \C{8}
 \end{tabular}
\caption{The axioms for conditionals.}
\label{tab:axioms_cnd}
\end{table}

Note that there is no elimination theorem here: due to the presence of unresolved propositional variables, conditionals cannot be removed from all closed terms.

Next, we introduce an operator that allows the observation of aspects
of the current state of a process graph. The central idea is that the
observable part of the state of a process graph is represented by a
proposition.  We introduce the
\emph{root-signal emission operator} $\rse{}{}$\index{root-signal
  emission}. A process expression $\rse{\phi}{x}$ represents the
process $x$ that shows the signal $\phi$ in its initial state. In
order to define this operator by operational rules, we need to define
an additional predicate on process expressions, namely
\emph{consistency}. $\mathit{Cons}(\langle p,v \rangle)$ will not hold
when the valuation of the root signal of $p$ is false. A step
$\step{a}$ can only be between consistent states, and a state can only
be accepting when it is consistent. Thus, if the effect of executing
the action $a$ is to set the value of the proposition $\phi$ to
$\mathit{false}$ (i.e., $\mathit{effect}(a,v)(\phi)=\mathit{false}$
for all valuations $v$), then the process expression
$a.(\rse{\phi}{p})$ can under no valuation execute action $a$.

The operational rules are defined in Fig. \ref{RSEsostable}. First, we
define the consistency predicate. Next, we find that the rule for action prefix, the rules for choice
and the second rule for $\seqc$ in Fig. \ref{GCsostable} require an
extra condition. The other rules of Table~\ref{GCsostable} can remain
unchanged. Finally, we give the operational rules of the root signal
emission operator.  To emphasise the difference between guarded
commands and root signal emission, process expression
$\phi :\rightarrow \emp$ is consistent under any valuation,
whereas $\rse{\phi} {\emp}$ is inconsistent under a
valuation that assigns $\mathit{false}$ to $\phi$. So, if the
effect of the action $a$ is to set the value of $\phi$ to
$\mathit{false}$, then $a.(\phi : \rightarrow \emp)$, in any
valuation, can execute an $a$-transition, whereas $a.\rse{\phi}{\emp}$
cannot.

 \begin{figure}[htb]
 \centering
 
 \begin{osrules}
 	\osrule*{}{\mathit{Cons}(\langle \dl,v \rangle)} \qquad \osrule*{}{\mathit{Cons}(\langle \emp,v \rangle)} \qquad \osrule*{}{\mathit{Cons}(\langle a.p,v \rangle)} \\
	\osrule*{\mathit{Cons}(\langle p,v \rangle) \; \mathit{Cons}(\langle q,v \rangle)}{\mathit{Cons}(\langle p + q,v \rangle)} \quad \osrule*{\mathit{Cons}(\langle p,v \rangle)}{\mathit{Cons}(\langle \phi : \rightarrow p,v \rangle)} \quad
     \osrule*{\mathit{Cons}(\langle p,v \rangle) \quad v(\phi) = \mathit{true}}{\mathit{Cons}(\langle \rse{\phi}{p},v \rangle)}  \\
   \osrule*{\mathit{Cons}(\langle p,v \rangle) & \langle p,v \rangle \not\downarrow}{\mathit{Cons}(\langle p \seqc q,v \rangle)}
  \quad
    \osrule*{\langle p,v \rangle \downarrow & \mathit{Cons}(\langle q,v \rangle)}{\mathit{Cons}(\langle p \seqc q,v \rangle) }  \\
    \osrule*{\mathit{Cons}(\langle p,v \rangle)}{\mathit{Cons}(\langle \mathit{NA}(p),v \rangle)} \qquad
 \osrule*{\mathit{Cons}(\langle p,v \rangle) & X\defeqn p}{\mathit{Cons}(\langle X,p \rangle)} \\
 \osrule*{\mathit{Cons}(\langle p,v' \rangle) \quad v'=\mathit{effect}(a,v)}{\langle a.p,v \rangle \step{a} \langle p,v' \rangle}  \quad
 \\
 \osrule*{\langle p,v \rangle \step{a} \langle p',v' \rangle \quad  \mathit{Cons}(\langle q,v \rangle)}{\langle p + q,v \rangle \step{a} \langle p',v'\rangle \;\; \langle q + p,v \rangle \step{a} \langle p',v'\rangle} \quad
  \osrule*{\langle p,v \rangle \downarrow \quad \mathit{Cons}(\langle q,v \rangle)}{\langle p + q,v \rangle \downarrow \;\; \langle q+p,v \rangle \downarrow} \\
  \osrule*{\langle p,v \rangle \step{a} \langle p',v' \rangle & \mathit{Cons}(\langle p' \seqc q,v' \rangle)}{\langle p \seqc q,v \rangle \step{a} \langle p' \seqc q,v' \rangle} \\ 
     \osrule*{\langle p,v \rangle \step{a} \langle p',v' \rangle \quad v(\phi) = \mathit{true}}{\langle \rse{\phi}{p},v \rangle \step{a} \langle p',v' \rangle} \qquad
    \osrule*{\langle p,v \rangle \downarrow  \quad v(\phi) = \mathit{true}}{\langle \rse{\phi} {p},v \rangle \downarrow} \\
 \end{osrules}

 \caption{Operational rules for root-signal emission ($a \in \mathcal{A} \cup \{\tau\}$).}
 \label{RSEsostable} %
\end{figure}

We again have a stateless bisimulation, where two process expressions are related
iff any valuation that makes the root signal of one process expression
$\mathit{true}$ also makes the root signal of the other process expression
$\mathit{true}$ and for each such valuation, the process graphs of the
process expressions are stateless bisimilar.

\begin{defi} \label{def:statelessbisimilaritysignal}
A binary relation $R$ on the set of sequential process expressions with conditionals and signals is a \emph{stateless bisimulation} iff $R$ is
symmetric and for all process expressions $p$ and $q$ such that if $(p,q) \in R$:
\begin{enumerate}
\item\label{item:cons} If for some valuation $v$ we have $\mathit{Cons}(\langle p,v \rangle)$, then also $\mathit{Cons}(\langle q,v \rangle)$.
\item\label{item:step} If for some $a \in \mathcal{A} \cup \{\tau\}$ and valuations
  $v,v'$ we have $\langle p,v \rangle \step{a} \langle p',v' \rangle$,
  then there exists a process expression $q'$, such that $\langle q,v \rangle \step{a} \langle q',v' \rangle$, and $(p', q') \in R$.
\item\label{item:term} If for some valuation $v$ we have $\term{\langle p,v \rangle}$, then $\term{\langle q,v \rangle}$.
\end{enumerate}
The process expressions $p$ and $q$ are bisimilar (notation: $p \sbisim q$) iff there exists a stateless bisimulation $R$ such that $(p, q) \in R$.
\end{defi}

Again, stateless bisimulation is a congruence for sequential process expressions with conditionals and signals (see \cite{BBR10}), and we have the following additional axioms.

\begin{table}[htb]
\centering
\begin{tabular}{lcl@{\quad}l}
  $\rse{\mathit{true}}{x}$ & $=$ & $x$ & \SI{1} \\
  $\rse{\mathit{false}}{x}$ & $=$ & $\rse{\mathit{false}}{\dl}$ &
                                                                  \SI{2}\\
  $a.(\rse{\mathit{false}}{x})$ & $=$ & $\dl$ & \SI{3} \\
  $(\rse{\phi}{x}) +y$ & $=$ & $\rse{\phi}{(x+y)}$ & \SI{4} \\
  $\rse{\phi}{(\rse{\psi}{x})}$  & $=$ & $\rse{(\phi \wedge \psi)}{x} $ & \SI{5} \\
  $\phi : \rightarrow (\rse{\psi}{x}) $ & $=$ & $\rse{(\neg \phi \vee \psi)}{(\phi : \rightarrow x)} $ &
  \SI{6} \\
  $\rse{\phi}{(\phi : \rightarrow x)}$ &=& $\rse{\phi}{x}$ & \SI{7}\\
  $\rse{\phi}{(x \seqc y)}$ & $=$ & $ (\rse{\phi}{x}) \seqc y$ &
                                                                    \SI{8}\\
  $\rse{\phi}{\na{x}}$ & $=$ & $ \na{\rse{\phi}{x}} $ &
                                                                    \SI{9}\\
 \end{tabular}
\caption{The axioms for signals.}
\label{tab:axioms_sig}
\end{table}

It is tedious but straightforward to verify that the extended theory
is still sound. Below we shall prove first a head normal form theorem
and then also a ground-completeness theorem.

\begin{thm}[head normal form theorem] \label{thm:hnf}
  Let $\Delta$ be a guarded recursive specification.
  For every process expression $p$ there exists a natural number $n$,
  sequences of propositions $\phi_1,\dots,\phi_n$, actions
  $a_1,\dots,a_n$, and process expressions $p_1,\dots,p_n$, and
  propositions $\psi$ and $\chi$ such that
  \begin{equation}\label{eq:hnf}
    p = \sum_{i=1}^{n}\phi_i : \rightarrow a_i.p_i
          + \rse{\psi}{\chi: \rightarrow \emp}
  \end{equation}
  is derivable from the axioms in
  Tables~\ref{tab:axioms_tsp}--\ref{tab:axioms_sig} and the equations in $\Delta$ using
  equational logic. The process expression at the right-hand side of
  the equation is called a \emph{head normal form} of $p$, $\psi$ is
  called its \emph{root signal} and $\chi$ is called its
  \emph{acceptance condition}.
\end{thm}
\begin{proof}
  Let $p$ be a process expression. Without loss of generality we may
  assume that $p$ is guarded, for we can replace every unguarded
  occurrence of a process identifier in $p$ by the right-hand side of
  its defining  equation in $\Delta$ which, since $\Delta$ is guarded, results in
  process expression with one fewer unguarded occurrence of a process
  identifier. We now proceed by induction on the structure of $p$.

  Note that $p$ cannot be itself a process identifier, since then $p$
  would not be guarded.

  If $p=\dl$, then
    $p=\dl+\dl =\dl +\rse{\mathit{true}}{\mathit{false}:\rightarrow\emp}$
  by axioms \A3, \C2 and \SI1{}.
  Note that $\dl +\rse{\mathit{true}}{\mathit{false}:\rightarrow\emp}$
  fits the shape of the right-hand side of Equation~\eqref{eq:hnf}: we take $n=0$ and we let
  $\psi = \mathit{true}$ and $\chi = \mathit{false}$.
 
  If $p=\emp$, then by \A6, \A1, \C1 and \SI1 we have
  $p=\emp+\dl=\dl+\emp=\dl + \rse{\mathit{true}}{\mathit{true}: \rightarrow\emp}$,
  which fits the shape of the right-hand side of Equation~\eqref{eq:hnf} if we
  take $n=0$, $\psi=\mathit{true}$ and $\chi=\mathit{true}$.

  If $p=a.p'$, then by \A6, \C2 and \SI1{} we have that
  $p=a.p'+\dl=a.p' + \rse{\mathit{true}}{\mathit{false}: \rightarrow \emp}$.
  
  Suppose, now, by way of induction hypothesis, that
  \begin{gather*}
    p' = \sum_{i=1}^{m}\phi_i' : \rightarrow a_i'.p_i'
          + \rse{\psi'}{\chi': \rightarrow \emp}\\
    \intertext{and}
    p'' = \sum_{i=1}^{n}\phi_i'' : \rightarrow a_i''.p_i''
          + \rse{\psi''}{\chi'': \rightarrow \emp}\enskip.
  \end{gather*}

  First note that if $p=p'+p''$, then, by \A{1}--\A{3}, \C3, \SI{4} and \SI{5}, we get that
  \begin{equation*}
    p= \sum_{i=1}^{m+n}\phi_i : \rightarrow a_i.p_i
          + \rse{(\psi'\wedge\psi'')}{(\chi'\vee\chi''): \rightarrow \emp}\enskip,
  \end{equation*}
  where
  $\phi_i=\phi_i'$, $a_i=a_i'$ and $p_i=p_i'$ if $1\leq i \leq m$, and
  $\phi_i=\phi_i''$, $a_i=a_{i-m}''$ and $p_i=p_{i-m}''$ if $m < i \leq m+n$.
  
  Next, we consider $p=p'\seqc p''$. By axioms \A{1}, \NA{1} (used if
    $m=0$ or $n=0$), \NA{3}, \NA{4}, \C7, \SI{4} and \SI{8} we have
   \begin{gather*}
    p' = \na{\rse{\psi'}{\sum_{i=1}^{m}\phi_i' : \rightarrow a_i'.p_i'}}
          +\chi': \rightarrow \emp\\
    \intertext{and}
    p'' = \na{\rse{\psi''}{\sum_{i=1}^{n}\phi_i'' : \rightarrow a_i''.p_i''}}
          + \chi'': \rightarrow \emp\enskip.
  \end{gather*}
  Hence, by \C8, we have
  \begin{equation*}
    p=p'\seqc p''=
      \na{\rse{\psi'}{\sum_{i=1}^{m}\phi_i' : \rightarrow a_i'.p_i'}}\seqc p''
      + (\chi'\wedge\chi''): \rightarrow \emp \enskip.
  \end{equation*}
  The root signal can be transferred to the other summand
  with applications of \SI8, \SI9, \SI4 and  \A1:
  \begin{equation*}
    p=
      \na{\sum_{i=1}^{m}\phi_i' : \rightarrow a_i'.p_i'}\seqc p''
      + \rse{{\psi'}}{(\chi'\wedge\chi''): \rightarrow \emp} \enskip.
  \end{equation*}
  We then find a head normal form for $p$ with applications of \A{11},
  \C6, \NA3, \C7, and \A5:
   \begin{equation*}
    p=
      \sum_{i=1}^{m}\phi_i' : \rightarrow a_i'.(p_i'\seqc p'')
      + \rse{{\psi'}}{(\chi'\wedge\chi''): \rightarrow \emp} \enskip.
  \end{equation*}
 
   If $p=\na{p'}$, then \NA{1}--\NA{4}, \A1,
   \A6, \C7, \SI9, \C2 and \C4 we have
  \begin{equation*}
    p=\na{p'}=\sum_{i=1}^{m}\phi_i' : \rightarrow a_i'.p_i'
    + \rse{\psi'}{\chi': \rightarrow \dl}
      =\sum_{i=1}^{m}\phi_i' : \rightarrow a_i'.p_i'
    + \rse{\psi'}{\mathit{false}: \rightarrow \emp} \enskip.
  \end{equation*}

  Suppose that $p=\phi : \rightarrow p'$, with the head normal form of
  $p'$ as given above. Let us first consider the case that $m>0$, then
  by \C5, \SI6 and \C4 we get the following head normal form:
  \begin{equation*}
      p = \sum_{i=1}^{m}(\phi \wedge \phi_i') : \rightarrow a_i'.p_i'
          + \rse{(\neg\phi \vee \psi')}{(\phi \wedge \chi'): \rightarrow
            \emp} \enskip.          
  \end{equation*}
  If $m=0$, then the same head normal form for
  $p$ can be obtained, since by \C2 and \C4 we have
  \begin{multline*}
    \phi :\rightarrow \sum_{i=1}^{m}\phi_i' : \rightarrow a_i'.p_i' =
     \phi :\rightarrow \dl = \phi : \rightarrow (\mathit{false}:
     \rightarrow \dl) =  \\ \mathit{false}: \rightarrow \dl =\dl =
     \sum_{i=1}^{m}(\phi \wedge \phi_i') : \rightarrow a_i'.p_i'\enskip.
   \end{multline*}

   Suppose that $p = \rse{\phi}{p'}$, again with the head normal form
   of $p'$ as given above. Then by \A{1}, \SI4 and \SI5 we obtain the
   following head normal form:
  \begin{equation*}
    p' = \sum_{i=1}^{m}\phi_i' : \rightarrow a_i'.p_i'
          + \rse{(\phi\wedge\psi')}{\chi': \rightarrow \emp}\qedhere
  \end{equation*}
\end{proof}

There may be some redundancy in a head normal form. Consider, for
instance, the head normal form
\begin{equation*}
  p = \dl + P :\rightarrow a.\emp + P:\rightarrow b.\rse{(\neg P)}{\emp}  + (\neg P):
  \rightarrow c.\emp + \rse{P}{\mathit{true}:\rightarrow \emp}\enskip.
\end{equation*}
Then the summand $\neg P:\rightarrow c.\emp$ does not contribute any behaviour, since
there does not exist a valuation that is consistent with the root
signal $P$ and at the same time makes condition $\neg P$ evaluate to
$\mathit{true}$. Moreover, if the $\mathit{effect}$ function satisfies, e.g., the
property that  $\mathit{effect}(b,v)(P)=v(P)$ for all valuations $v$, then
the effect of $b$ is not consistent with the root signal of $\rse{\neg
  P}{\emp}$, and so  the summand $P:\rightarrow b.\rse{(\neg P)}{\emp}$
does not contribute any behaviour.

\begin{defi} \label{def:rhnf}
A head normal form
  $\sum_{i=1}^{n}\phi_i : \rightarrow a_i.p_i
  + \rse{\psi}{\chi: \rightarrow \emp}$
is \emph{reduced} if, for every $1\leq i \leq n$, there exists a
valuation $v$ such that $v(\phi_i\wedge\psi)=\mathit{true}$ and
$\mathit{Cons}(\langle p_i,\mathit{effect}(a_i,v)\rangle)$.
\end{defi}

In the remainder we shall only be interested in a very specific type
of $\mathit{effect}$ function. Let us denote by $v_{\mathit{true}}$
the valuation that assigns $\mathit{true}$ to every propositional
variable. We say that the $\mathit{effect}$ function has the
\emph{reset property} if $\mathit{effect}(a,v)=v_{\mathit{true}}$ for
every action $a$ and every valuation $v$. If the $\mathit{effect}$
function has the reset property, then the axiom in
Table~\ref{tab:axioms_effect} is valid.

\begin{table}[thb]
\centering
\begin{tabular}{lcl@{\quad}l}
  $a.\rse{(\neg P_1\vee\dots\vee \neg P_k)}{x}$ & $=$ & $\dl$ & \textrm{R}
 \end{tabular}
\caption{The axiom for the $\mathit{effect}$ function with the reset
  property ($a\in\mathcal{A}\cup\{\tau\}$ and $P_1,\dots,P_k$ is any
  sequence of propositional variables).}
\label{tab:axioms_effect}
\end{table}

\begin{lem} \label{rhnf}
  For every head normal form $p$ with a consistent root signal there
  exists a reduced head normal form $q$ such that $p=q$ is derivable
  from the axioms in Tables~\ref{tab:axioms_tsp}--\ref{tab:axioms_effect}.
\end{lem}
\begin{proof}
  Let
    $p=\sum_{i=1}^{n}\phi_i : \rightarrow a_i.p_i + \rse{\psi}{\chi:
      \rightarrow \emp}$.
  If for some $1 \leq i \leq n$, there does not exist a valuation $v$ such that
  $\phi_i\wedge\psi=\mathit{true}$, then $\phi_i \wedge \psi =
  \mathit{false}$, so by axioms \SI{7}, \C4, 
  and \C2 we have  that
  \begin{multline*}
    \rse{\psi}{(\phi_i:\rightarrow a_i.p_i)} =
    \rse{\psi}{(\psi:\rightarrow(\phi_i: \rightarrow a_i.p_i))} =
    \rse{\psi}{(\phi_i\wedge\psi): \rightarrow a_i.p_i} \\ =
    \rse{\psi}{(\mathit{false}:\rightarrow a_i.p_i)} = \rse{\psi}{\dl} =
    \dl
  \enskip.
  \end{multline*}
  Hence, by \SI4, \A{1}--3 and \A6 such a summand can be eliminated from
  the head normal form.

  It remains to argue that if $\mathit{Cons}(\langle
  p_i,v_{\mathit{true}}\rangle)$ does not hold, then the summand
    $\phi_i:\rightarrow
  a_i.p_i$ can be eliminated.
  Note that by Theorem~\ref{thm:hnf}, $p_i$ is derivably equal to a head normal
  form, say
  \begin{equation*}
    p_i = \sum_{j=1}^{n_i}\phi_{i,j} :\rightarrow a_{i,j}.p_{i,j} +
    \rse{\psi_i}{\chi_i: \rightarrow \emp}\enskip.
  \end{equation*}
  From the rules in Figure~\ref{RSEsostable} it follows that we have
  $\mathit{Cons}(\langle p_i, v_{\mathit{true}}\rangle)$ if, and only
  if, $v_\mathit{true}(\psi_i)=\mathit{true}$. Suppose that
  $v_{\mathit{true}}(\psi_i)=\mathit{false}$. Since $\psi_i$ is an element of a Boolean
  algebra $\mathcal{B}$ that is freely generated by the propositional
  variables, there exists a finite sequence $P_1,\dots,P_k$ of propositional
  variables such that, for every valuation $v$, $v(\psi_i)$ is
  completely determined by the truth values it assigns to
  these propositional variables. From
  $v_{\mathit{true}}(\psi_i)=\mathit{false}$, it then follows that for
  all valuations $v$ such that $v(\psi_i)=\mathit{true}$ there exists
  $1\leq i \leq k$ such that $v(P_i)=\mathit{false}$.
  Therefore, $\psi_i = \psi_i \wedge (\neg P_1\vee\dots\vee \neg
  P_k)$. Applications of axioms \SI5 and $\textrm{R}$ then yield
  \begin{equation*}
    \phi_i:\rightarrow a_i.p_i =
      \phi_i: \rightarrow a_i.\rse{(\neg P_1\vee\dots\vee\neg
        P_k)}{p_i} = \dl.
  \end{equation*}
  We conclude that if $\mathit{Cons}(\langle p_i,v_{\mathit{true}}\rangle)$ does
  not hold, then the summand $\phi_i:\rightarrow a_i.p_i$ can be
  eliminated by \A6.
\end{proof}
 
Below, we shall prove that, assuming an $\mathit{effect}$ function
with the reset property, the axioms in
Tables~\ref{tab:axioms_tsp}--\ref{tab:axioms_effect} are complete for
recursion-free process expressions. The proof will use induction on
the depth of the process expressions involved.
\newcommand{\depth}[1]{\ensuremath{\mathalpha{|#1|}}}
\begin{defi}\label{def:depth}
  We define the \emph{depth} $\depth{p}$ of a process expression $p$ by
  \begin{multline*}
    \depth{p} = \sup\{
      n \mid \exists p_0,\dots,p_n.\exists v_1,\dots,v_n.\exists
      a_1,\dots,a_n.\exists v_1',\dots,v_n'.\\
        p=p_0\ \&\ \forall_{1\leq i < n}. \langle p_i,v_i\rangle \step{a_{i+1}} \langle p_{i+1},v_i'\rangle
        \}\enskip.
  \end{multline*}
\end{defi}
From the operational semantics and the definition of stateless
bisimilarity it follows that if $p\sbisim q$, then
$\depth{p}=\depth{q}$. Moreover, if $p$ is recursion-free, then
$\depth{p}$ is finite. The following lemma facilitates our
ground-completeness proof to proceed by induction on depth.

\begin{lem} \label{lem:depth}
  Let $p=\sum_{i=1}^{n}\phi_i : \rightarrow a_i.p_i + \rse{\psi}{\chi:
    \rightarrow \emp}$ be a recursion-free and reduced
  head normal form and suppose that $\mathit{effect}$ has the reset
  property; then $\depth{p} > \depth{p_i}$ for all $1\leq i \leq n$.
\end{lem}
\begin{proof}
  Let $1\leq i \leq n$. Then since $p$ is reduced, there exists a
  valuation $v$ such that $v(\phi_i\wedge\psi)\neq
  \mathit{false}$. Since $\mathit{effect}$ has the reset property, we
  have that $\mathit{effect}(a_i,v)=v_{\mathit{true}}$, and since $p$
  is reduced we have that $\mathit{Cons}(\langle
  p_i,v_{\mathit{true}}\rangle)$.
  Therefore $\langle p,v\rangle \step{a_i}\langle
  p_i,v_{\mathit{true}}\rangle$.
  Since the depth of recursion-free processes is finite, it follows that $\depth{p}>\depth{p_i}$.
\end{proof}

We can now establish that our axiomatisation is ground-complete.
\begin{thm}[Ground-completeness] \label{thm:completeness}
  Suppose that the $\mathit{effect}$ function has the reset property.
  For all recursion-free process expressions $p$ and $p'$, if
  $p\sbisim p'$ then $p=p'$ is derivable from the axioms in
  Tables~\ref{tab:axioms_tsp}--\ref{tab:axioms_effect} using
  equational logic.
\end{thm}
\begin{proof}
   We proceed by induction on the sum of the depths of $p$ and $q$. By
   Theorem~\ref{hnf} and Lemma~\ref{rhnf} we may assume that $p$ and
   $p'$ are reduced head normal forms: 
   \begin{equation*}
      p=\sum_{i=1}^{n}\phi_i : \rightarrow a_i.p_i + \rse{\psi}{\chi:
        \rightarrow \emp}
    \end{equation*}
    and
    \begin{equation*}
      p'=\sum_{j=1}^{n'}\phi_j' : \rightarrow a_j'.p_j' + \rse{\psi'}{\chi':
        \rightarrow \emp}
    \enskip.
    \end{equation*}
    First note that by condition~\ref{item:cons} of
    Definition~\ref{def:statelessbisimilaritysignal} we have that
    $\psi=\psi'$ and by condition~\ref{item:term} we have that
    $\chi=\chi'$. Hence, it remains to establish that
      $p' = p' +  \phi_i : \rightarrow a_i.p_i$
      for all $1\leq i \leq n$.

    Note that if $v$ is some valuation such that
    $v(\phi_i)=\mathit{true}$, then $\langle p, v\rangle\step{a_i}
    \langle p_i,v_{\mathit{true}}\rangle$. It then
    follows by condition~\ref{item:step} of
    Definition~\ref{def:statelessbisimilaritysignal} that there exists
    $1\leq j \leq n'$ such that
    $\langle p',v\rangle\step{a_j} \langle
    p_j',v_{\mathit{true}}\rangle$, $a_j=a_i$ and $p_i \sbisim p_j'$,
    and hence $v(\phi_j')=\mathit{true}$.
    Define
    \begin{equation*}
      J' = \{ 1\leq j \leq   n' \mid a_j=a_i \ \&\ p_j \sbisim p_i\}
      \enskip.
    \end{equation*}
    We have just argued that $v(\phi_i)=\mathit{true}$ implies that
    there exists $j\in J'$ such that $v(\phi_j')=\mathit{true}$, and
    hence
    \begin{equation}  \label{eq:condition}
      \bigvee_{j\in J'}\left(\phi_i \wedge \phi_j'\right) =
      \phi_i\enskip.
    \end{equation}
    Using that
      $\phi_j' = (\phi_i \wedge \phi_j') \vee
      (\phi_i\wedge\neg\phi_j')$,
    note that by axioms \C3 and \A3 we have
    \begin{equation*}
      \phi_j': \rightarrow a_j'.p_j' =
         \phi_j':\rightarrow a_j'.p_j' + (\phi_i \wedge
         \phi_j'):\rightarrow a_j'.p_j'\enskip,
   \end{equation*}
   and so we have
    \begin{equation*}
      p' = p' + \sum_{j\in J'} (\phi_i \wedge \phi_j'):\rightarrow a_j'.p_j'\enskip.
   \end{equation*}
   By Lemma~\ref{lem:depth} we have $\depth{p_i}<\depth{p}$ and
   $\depth{p_j'}<\depth{p'}$, so by the induction hypothesis
   $p_j'=p_i$ is derivable from the axioms in
   Tables~\ref{tab:axioms_tsp}--\ref{tab:axioms_effect} .
   It follows that
   \begin{equation*}
      p' = p' + \sum_{j\in J'} (\phi_i \wedge \phi_j'):\rightarrow a_i.p_i\enskip.
   \end{equation*}
   Finally, we apply \C3 and Eqn.~\eqref{eq:condition} to get
   \begin{equation*}
      p' = p' + \phi_i:\rightarrow a_i.p_i\enskip.\qedhere
   \end{equation*}
\end{proof}

The information given by the truth of the signals allow to determine
the truth of some of the guarded commands. In this way, we can give a
semantics for process expressions and specifications as regular process graphs, leaving out the valuations. 

\begin{defi} \label{tss}
Let $p, q$ be sequential process expressions with signals and conditions, let $a \in \mathcal{A} \cup \{\tau\}$ and suppose the root signal of $t$ is not \emph{false}.
\begin{itemize}
\item $p \step{a} q$ iff for all valuations $v$ such that $\mathit{Cons}(\langle p,v \rangle)$ we have $\langle p,v \rangle \step{a} \langle q,\mathit{effect}(a,v) \rangle$,
\item $p \downarrow$ iff for all valuations $v$ such that $\mathit{Cons}(\langle p,v \rangle)$ we have $\langle p,v \rangle \downarrow$.
\end{itemize}
\end{defi}

  The above definition makes all undetermined guarded commands
  \emph{false}, as is illustrated in the following example.
\begin{exa} \label{exa:falseguards}
  Let $P$ be a proposition variable,
  and let $p = P: \rightarrow a.\emp$ and $q = P:\rightarrow
  b.\emp$. Note that we have $\mathit{Cons}(\langle p,v\rangle)$ for
  all valuations $v$ (since $p$ does not emit any signal), but
  $\langle p,v\rangle \step{a} \langle \emp,v'\rangle$ only if
  $v(P)=\mathit{true}$. Hence, $p$ does not have any outgoing
  transitions according to Definition~\ref{tss}. By the same
  reasoning, also $q$ does not have any outgoing
  transitions. Therefore, $p$ and $q$ are bisimilar with respect to
  the transition relation induced on them by Definition~\ref{tss}.
\end{exa}

When two process expressions are stateless bisimilar, then they are also bisimilar
with respect to the transition relation induced on them by the
Definition~\ref{tss}. The converse, however, does not hold: although
the process expressions $p$ and $q$ in Example~\ref{exa:falseguards}
are bisimilar, they are not stateless bisimilar.

To illustrate the interplay of root signal emission and guarded command, and to show how nondeterminism can be dealt with, we give the following example.

\begin{figure}[htb]
\begin{center}
\begin{tikzpicture}[->,>=stealth',node distance=1.5cm, node
  font=\footnotesize, state/.style={ellipse, draw, minimum size=.5cm,inner sep=0pt}]
  \node[state,initial,initial text={},initial where=left] (s0) {};
  \node[state] [above right=of s0,yshift=-.5cm,minimum width=1cm] (s1) {$\mathit{tails}$};
  \node[state] [below right=of s0,yshift=.5cm,minimum width=1cm] (s2) {$\mathit{heads}$};
  \node[state,accepting] [right of=s2,node distance=2cm] (s3) {};
  
  \path[->]
    (s0) edge node[left,align=right,yshift=0.2cm] {$\mathit{toss}$} (s1)
    (s0) edge node[left,align=right,yshift=-0.2cm] {$\mathit{toss}$} (s2)
    (s1) edge node[right] {$\mathit{toss}$} (s2)
    (s1) edge[loop right] node[right] {$\mathit{toss}$} (s1);
  \path[->]
    (s2) edge node[above] {$\mathit{hurray}$} (s3);
\end{tikzpicture}
\end{center}
\caption{The process graph associated with the specification of a coin toss.}\label{fig:cointoss}
\end{figure}

\begin{exa} \label{exa:cointoss}
A coin toss can be described by the following process expression: 
\[ T \defeqn \mathit{toss}.(\rse{\mathit{heads}}{\emp}) + \mathit{toss}.(\rse{\mathit{tails}}{\emp}), \]
where the \emph{effect} function has the reset property, i.e.\ for every
valuation $v$ we have
\begin{equation*}
  \mathit{effect}(\mathit{toss,v})(\mathit{heads}) =
  \mathit{effect}(\mathit{toss,v})(\mathit{tails}) = \mathit{true}\enskip.
\end{equation*}
  Now, consider the expression 
 \[S \defeqn T \seqc (\mathit{heads} : \rightarrow
   \mathit{hurray}.\emp + \mathit{tails} : \rightarrow S). \]

   Then $S$ represents the process of
   tossing a coin until heads comes up, and its process graph is
   shown in Figure~\ref{fig:cointoss}, as we shall now explain. Let
     \[\mathit{Heads}= (\rse{\mathit{heads}}{\emp}) \seqc (\mathit{heads} : \rightarrow
   \mathit{hurray}.\emp + \mathit{tails} : \rightarrow S)\]
   and let
    \[\mathit{Tails}=(\rse{\mathit{tails}}{\emp}) \seqc (\mathit{heads} : \rightarrow
      \mathit{hurray}.\emp + \mathit{tails} : \rightarrow
      S)\enskip.\]
   To see that $S\step{\mathit{toss}}\mathit{Tails}$ and $S
   \step{\mathit{toss}}\mathit{Heads}$, note that
   $\langle S,v\rangle \step{\mathit{toss}} \langle \mathit{Heads},
   \mathit{effect}(\mathit{toss},v)\rangle$
   and
   $\langle S,v\rangle \step{\mathit{toss}} \langle \mathit{Heads},
   \mathit{effect}(\mathit{toss},v)\rangle$ for every valuation $v$.

   To see that $\mathit{Tails}\step{\mathit{toss}}\mathit{Tails}$ and
   $\mathit{Tails}\step{\mathit{toss}}\mathit{Heads}$, first observe that $\mathit{Cons}(\langle \mathit{Tails},v\rangle)$ if, and
   only if, $v(\mathit{tails})=\mathit{true}$, and then note that for
   all such valuations $v$ we, indeed, have
        $\langle \mathit{Tails},v\rangle \step{toss}
     \langle\mathit{Tails},\mathit{effect}(\mathit{toss},v)\rangle$
     and
     $\langle \mathit{Tails},v\rangle \step{toss}
     \langle\mathit{Heads},\mathit{effect}(\mathit{toss},v)\rangle$.

   It is instructive to see why we do \emph{not} have that
   $\mathit{Tails}\step{\mathit{hurray}}\emp$. This is because
   if $v$ is a valuation that satisfies  $v(tails)=\mathit{true}$ and
   $v(\mathit{heads})=\mathit{false}$,  then we also have
   $\mathit{Cons}(\langle \mathit{Tails},v\rangle)$, whereas
   $\langle \mathit{Tails},v \rangle\nstep{\mathit{hurray}}\langle
   \emp,\mathit{effect}(\mathit{hurray},v)\rangle$.

   Finally, to see that $\mathit{Heads}\step{\mathit{hurray}}\emp$,
   first observe that $\mathit{Cons}(\langle \mathit{heads},v\rangle)$ if, and
   only if, $v(\mathit{heads})=\mathit{true}$, and then note that, indeed,
   $\langle \mathit{Heads},v \rangle\step{\mathit{hurray}}\langle
   \emp,\mathit{effect}(\mathit{hurray},v)\rangle$ for all such
   valuations $v$.
\end{exa}

\section{The full correspondence}

We prove that signals and conditions make it possible to find a guarded sequential specification with the same process as a given pushdown automaton.

\begin{thm} \label{th:correspondence}
For every  pushdown automaton there is a guarded sequential recursive
specification with signals and conditions with the same process.
\end{thm}
\begin{proof}
Let $M = (\mathcal{S},\mathcal{A},\mathcal{D},{\rightarrow},{\uparrow},{\downarrow})$ be a pushdown automaton. 
For every state $s \in \mathcal{S}$ we have a propositional variable
$\mathit{state}(s)$. We assume an $\mathit{effect}$ function with the
reset property, to make sure that the execution of an
  action from a consistent state always results in a consistent
  state. Note that it is the interplay between the emitted root signal and
  the guards that ensures appropriate continuations (cf.\ also Example~\ref{exa:cointoss}).We proceed to define the recursive specification with initial identifier $X$ and additional identifiers $X_{\epsilon}$ and $\{X_d \mid d \in \mathcal{D}\}$.

\begin{itemize}
\item If $M$ does not contain any transition of the form
${\uparrow} \xrightarrow{a[\epsilon/x]} s$, and the initial state is final, we can take $X \defeqn \emp$ as the specification, and we do not need the additional identifiers.
\item If $M$ does not contain any transition of the form
${\uparrow} \xrightarrow{a[\epsilon/x]} s$, and the initial state is not final, we can take $X \defeqn \dl$ as the specification, and we do not need the additional identifiers.
\item Otherwise, there is some transition ${\uparrow}
  \xrightarrow{a[\epsilon/x]} s$ in $M$. For each such transition, add
  a summand
  \begin{equation*}
      a.(\rse{\mathit{state}(s)}{\alpha_x \seqc X_{\epsilon}})
    \end{equation*}
    to the equation of $X$. The effect $v$ of $a$ in any valuation  satisfies $v(\mathit{state}(s))=\mathit{true}$. Besides these summands, add a summand $\emp$ iff the initial state of $M$ is final. Notice that no restriction on the initial valuation is necessary.
\item Next, the equation for the added identifier $X_d$ has a summand 
  \begin{equation*}
     \mathit{state}(s) : \rightarrow a.(\rse{\mathit{state}(t)}{\alpha_x})
   \end{equation*}
   for each transition $s \xrightarrow{a[d/x]} t$, for every $s,t \in \mathcal{S}$. The effect $v$ of $a$ in any valuation satisfies $v(\mathit{state}(t))=\mathit{true}$.

Finally, the equation for the added identifier $X_d$ has a summand
$\mathit{state}(s) : \rightarrow \emp$ whenever $s \in {\downarrow}$.
\item Lastly, the equation for the added identifier $X_{\epsilon}$ has a summand
  \begin{equation*}
     \mathit{state}(s) : \rightarrow a.(\rse{\mathit{state}(t)}{\alpha_x \seqc X_{\epsilon} })
 \end{equation*} 
        for each transition $s \xrightarrow{a[\epsilon /x]} t$, for every $s,t \in \mathcal{S}$. The effect $v$ of $a$ in any valuation satisfies $v(\mathit{state}(t))=\mathit{true}$.

Finally, the equation for the added identifier $X_{\epsilon}$ has a summand
$\mathit{state}(s) : \rightarrow \emp$ whenever $s \in {\downarrow}$.
\end{itemize}
Then the bisimulation relation will relate the initial states, and will relate every state $(s,x)$ of the process graph of the pushdown automaton to state $\rse{\mathit{state}(s)}{\alpha_x \seqc X_{\epsilon}}$ of the process graph of the constructed specification. We see that the process graph of the pushdown automaton has a step $(s,dy) \step{a} (s,xy)$ (coming from a step $s \xrightarrow{a[d/x]} t$ in the automaton) if, and only if, $\rse{\mathit{state}(s)}{X_d \seqc \alpha_y \seqc X_{\epsilon}} \step{a} \rse{\mathit{state}(t)}{\alpha_x \seqc \alpha_y \seqc X_{\epsilon}}$ in the process graph of the specification. 

Also, the process graph of the pushdown automaton has a step $(s,\epsilon) \step{a} (s,x)$ (coming from a step $s \xrightarrow{a[\epsilon/x]} t$ in the automaton) if, and only if, $\rse{\mathit{state}(s)}{ X_{\epsilon}} \step{a} \rse{\mathit{state}(t)}{\alpha_x \seqc X_{\epsilon}}$ in the process graph of the specification.
\end{proof}

\begin{exa}
For the pushdown automaton in Fig. \ref{fig:pda2}, we find the following guarded recursive specification:
 \[ S = a.(\rse{\mathit{state}\!\uparrow}{A} \seqc (\mathit{state}\!\uparrow : \rightarrow S + \mathit{state}\!\downarrow : \rightarrow \emp)) + c.(\rse{\mathit{state}\!\downarrow}{\emp}) \]
 \[ A = \mathit{state}\!\downarrow : \rightarrow   b.(\rse{\mathit{state}\!\downarrow}{\emp}) \: + \]
 \[ \qquad  + \: \mathit{state}\!\uparrow : \rightarrow (a. (\rse{\mathit{state}\!\uparrow}{A} ; A) + b.(\rse{\mathit{state}\!\uparrow}{\emp}) + c.(\rse{\mathit{state}\!\downarrow}{A})). \]
\end{exa}

Finally, we are interested in the question whether the mechanism of signals and conditions is not \emph{too} powerful: can we find a pushdown automaton with the same process for any guarded sequential specification with signals and conditions? We will show the answer is positive. This means we recover the perfect analogue of the classical theorem: the set of processes given by pushdown automata coincides with the set of processes given by guarded sequential specifications with signals and conditions.

In order to prove this final theorem, we need another refinement of
Greibach normal forms, now in the presence of signals and
conditions.
Recall the head normal form of the previous section. Using Theorem~\ref{thm:hnf} and, if necessary, adding a finite number
of process identifiers with appropriate defining equations, we can write each equation in a guarded sequential specification with signals and conditions in Greibach normal form, as follows:
\[ X = \sum_{i=1}^{n} \phi_i : \rightarrow a_i.\alpha_i +
  \rse{\psi}{(\chi :\rightarrow \emp)} \qquad a_i \in {\Act} \cup
  \{\tau\},  X \in \Procids, \alpha_i \in \Procids^{*}\enskip. \]
Here, $\psi$ is the root signal of $X$, and $\chi$ the
acceptance condition of $X$, writing again
\[ X = \sum_{i=1}^{n} \phi_i : \rightarrow a_i.\alpha_i + \rse{\psi}{\dl}\qquad a_i \in {\Act} \cup \{\tau\},  X \in \Procids, \alpha_i \in \Procids^{*} \]
if $\psi\wedge\chi$ is \emph{false}.
We see that $X$ is accepting in any
state with a valuation $v$ such that $v(\psi\wedge\chi)=\mathit{true}$, and it shows intermediate acceptance
whenever there exists $1\leq i \leq n$ with $v(\psi \wedge \chi \wedge \phi_i) = \mathit{true}$.

We define the Acceptance Irredundant Greibach normal form as before, disregarding the remaining signals and conditions. 

 \begin{thm}
For every guarded sequential recursive specification with signals and
conditions there is a pushdown automaton with the same process.
\end{thm}
\begin{proof}
Suppose a finite guarded sequential specification with signals and
conditions is given. Without loss of generality we can assume this specification is in Acceptance Irredundant Greibach normal form, so every state of the specification is given by a sequence of identifiers that is acceptance irredundant. 

Because the specification is in AIGNF, it becomes easier to establish the acceptance and steps for any state of the process graph of the specification.
For any variable of the specification, we have $X \downarrow$ iff for
all valuations $v$ that make the root signal of $X$ true also make the
acceptance condition of $X$ true, and $X \step{a} \alpha$ iff for all valuations $v$ that make the root condition of $X$ also make the guard of $a$ true and $\mathit{Cons}(v,a)$.

Because of the acceptance irredundancy, it holds that if $X \downarrow$ then also $X \seqc \beta \downarrow$ for every reachable state of the specification. Also, if $X \step{a} \alpha$ then also $X \seqc \beta \step{a} \alpha \seqc \beta$ for every reachable state of the specification $X \seqc \beta$.

This means we can follow the proof of Theorem~\ref{theo3} exactly, and we are done.
\end{proof}
 
\section{Conclusion}
We looked at the classical theorem, that the set of languages given by
a pushdown automaton coincides with the set of languages given by a
context-free grammar. A language is an equivalence class of
process graphs modulo language
equivalence. A process is an equivalence class of process graphs modulo bisimulation.

This paper solves the question how we can characterize the set of processes given by a pushdown automaton. In the process setting, a context-free grammar is a process algebra with actions, choice, sequential composition and recursion. We need to limit to guarded recursion in the process setting. Starting out from the seminal process algebra BPA of \cite{BK84}, we need to add constants for the inactive and accepting process and for the inactive non-accepting (deadlock) process. We look at the candidate process algebra \TSP{} of \cite{BBR10}, but find it is not suitable: by means of a guarded recursive specification we find a process that cannot be the process of a pushdown automaton, and also, there is a process of a pushdown automaton that cannot be given by a guarded recursive specification.

If we replace the sequential composition operator of \TSP{} by a sequencing operator \cite{BLY17}, we get a better situation: every guarded specification yields the process of a pushdown automaton, but still, not the other way around, there are processes of pushdown automata that cannot be given by a guarded specification. We obtain a complete correspondence by adding a notion of state awareness, that allows to pass on some information during sequencing.

\medskip

\fbox{\parbox{.9\textwidth}{A process is defined by a pushdown
    automaton if and only if it is defined by a finite guarded recursive specification over a process algebra with actions, choice, sequencing, conditions and signals.}}

\medskip

We see that signals and conditions add expressive power to the process algebras considered, since a signal can be passed along the sequencing (or sequential composition) operator. If we go to the theory BCP, so without sequencing but with parallel composition, then we know from \cite{BB97} that value passing can be replaced by signal observation. We leave it as an open problem, whether or not signals and conditions add to the expressive power of BCP.

A pushdown automaton is a model of a computer with a memory in the form of a stack. A \emph{parallel pushdown automaton} is like a pushdown automaton, but with a memory in the form of a multiset (a set, where the elements have a multiplicity, telling how often they occur). It is interesting to compare the set of processes defined by a parallel pushdown automaton with the \emph{basic parallel} processes, the processes that are defined by a guarded recursive specification over the language, where sequencing is replaced by parallel composition. We conjecture that every process defined by a basic parallel recursive specification can also be defined by a parallel pushdown automaton, and in the other direction that a process defined by a one-state parallel pushdown automaton can also be defined by a basic parallel recursive specification. In order to recover a full correspondence, we need to add an extra feature. We speculate this extra feature can be communication between parallel processes, without abstraction. Some details can be found in \cite{vT11}.

This paper contributes to our ongoing project to integrate automata theory and process theory. As a result, we can present the foundations of computer science using a computer model with interaction. Such a computer model relates more closely to the computers we see all around us.

\section*{Acknowledgement}
Thanks to the anonymous referees we were able to correct some errors and make some clarifications, improving on \cite{BCL21}.

\bibliographystyle{alphaurl}
\bibliography{special}

\end{document}